\documentclass[runningheads]{llncs}
\usepackage{graphicx}
\usepackage[colorlinks,linkcolor=black,citecolor=blue]{hyperref}

\usepackage{sansmathaccent}
\usepackage{amsmath}
\usepackage{amssymb}
\usepackage{pifont}
\usepackage{hhline}
\usepackage{comment}
\usepackage[caption=false]{subfig}
\pdfmapfile{+sansmathaccent.map}
\usepackage{refcount} 

\usepackage[T1]{fontenc}
\IfFileExists{luximono.sty}{\usepackage[scaled=0.9]{luximono}}{\usepackage[scaled=0.81]{beramono}}
\usepackage{listings, silver}
\usepackage{color}
\usepackage{soul}
\usepackage{makecell}
\usepackage{multirow}
\usepackage{dirtree}
\usepackage{xspace}
\definecolor{light-gray}{gray}{0.9}
\definecolor{darkgreen}{rgb}{0.0,0.7,0.0}

\definecolor{GhostOpColor}{rgb}{0.0, 0.1, 1.0}
\newcommand\ghostopcode[1]{{\color{GhostOpColor}\bfseries\code{#1}}}

\usepackage{tikz}

\newcommand\code[1]{{\small\texttt{#1}}}
\newcommand\secref[1]{Sec.~\ref{#1}}
\newcommand\appref[1]{App.~\ref{#1}}
\newcommand\happref[1]{App.~#1}

\newcommand\defref[1]{Def.~\ref{#1}}

\newcommand\figref[1]{Fig.~\ref{#1}}
\newcommand\thmref[1]{Thm.~\ref{#1}}

\newcommand{\wrt}{{{w.r.t.\@}}}
\newcommand{\eg}{{{e.g.\@}}}
\newcommand{\ie}{{{i.e.\@}}}

\newcommand{\etal}{{\it{et al.\@}}}

\usepackage[normalem]{ulem}

\def\nocolour{ }

\newcommand{\soutifcolour}[1]{\ifdefined\nocolour{}\else{\protect\sout{#1}}\fi}

\newcommand\todo[1]{\ifdefined\nocolour{}\else{\textcolor{red}{TODO: #1}}\fi}

\newcommand{\peter}[1]{\ifdefined\nocolour{#1}\else{\color{orange}{#1}}\fi}

\newcommand{\pout}[1]{\peter{{\soutifcolour{#1}}}}

\newcommand{\as}[1]{\ifdefined\nocolour{{#1}}\else{\color{red}{#1}}\fi}

\newcommand{\asout}[1]{\as{\soutifcolour{#1}}}

\newcommand{\thibault}[1]{\ifdefined\nocolour{{#1}}\else{\color{blue}{#1}}\fi}
\newcommand{\tfootnote}[1]{\ifdefined\nocolour{}\else\thibault{\footnote{\thibault{THIBAULT: #1}}}\fi}
\newcommand{\tout}[1]{\thibault{\soutifcolour{#1}}}

\newcommand{\gaurav}[1]{\ifdefined\nocolour{#1}\else{\color{darkgreen}{#1}}\fi}
\newcommand{\gfootnote}[1]{\ifdefined\nocolour{}\else\gaurav{\footnote{\gaurav{GAURAV: #1}}}\fi}
\newcommand{\gout}[1]{\gaurav{{\soutifcolour{#1}}}}

\newcommand{\wand}{\ensuremath{\mathbin{-\!\!*}}}
\newcommand{\cwand}{\ensuremath{\mathbin{-\!\!*}_c}}
\newcommand{\twand}{\ensuremath{\mathbin{-\!\!*}_{\mathcal{T}}}}

\newcommand{\bkeyword}[1]{\texttt{\bfseries #1}}
\newcommand{\bcommand}[2]{\bkeyword{#1}\;#2}
\newcommand{\bassert}[1]{\bcommand{assert}{#1}}
\newcommand{\bfold}[1]{\bcommand{fold}{#1}}
\newcommand{\bunfold}[1]{\bcommand{unfold}{#1}}
\newcommand{\bapply}[1]{\bcommand{apply}{#1}}

\definecolor{darkred}{rgb}{0.55, 0.0, 0.0}

\newcommand*{\pointsto}[2]{\ensuremath{{#1} \mapsto {#2}}}

\newcommand{\Assign}[2]{\ensuremath{#1\gets#2}}

\makeatletter
\def\CAS#1#2#3#4{\@ifnextchar\bgroup {\CASAssign{#1}{#2}{#3}{#4}}{\CASOp{#1}{#2}{#3}{#4}}}
\newcommand{\CASOp}[4]{\texttt{CAS}_{#1}(#2,#3,#4)}
\newcommand{\CASAssign}[5]{\Assign{#1}{\CASOp{#2}{#3}{#4}{#5}}}

\newcommand{\pc}{\textit{pc}}

\usepackage{esvect}

\usepackage{algorithm}
\usepackage{algorithmic}

\usepackage{enumitem}
\setlist{nolistsep}

\usepackage{prooftree}

\makeatletter

\newskip \point

\def \premisespacing{\quad \quad}

\def \RulePremisesNewlineMore[#1]#2.#3#4{\@ifnextchar\bgroup{\RulePremisesNewlineMore[#1]{#2}.{#3\premisespacing#4}}{\@ifnextchar.{\RulePremisesNewline[#1]{{\begin{array}{c}#2\\#3\premisespacing#4\end{array}}}}{\RuleMultiPremise[#1]{{\begin{array}{c}#2\\#3\end{array}}}{#4}}}}

\def \RulePremisesNewline[#1]#2.#3{\@ifnextchar\bgroup{\RulePremisesNewlineMore[#1]{#2}.{#3}}{\@ifnextchar.{\RulePremisesNewline[#1]{{\begin{array}{c}#2\\#3\end{array}}}}{\RuleMultiPremise[#1]{#2}{#3}}}}

\def \RuleMultiPremise[#1]#2#3{\@ifnextchar\bgroup{\RuleMultiPremise[#1]{#2\premisespacing#3}}{\@ifnextchar.{\RulePremisesNewline[#1]{#2\premisespacing#3}}{\prooftree #2\justifies#3 \using{#1}\endprooftree}}}

\def \RuleWithName[#1]#2{\@ifnextchar\bgroup {\RuleMultiPremise[#1]{#2}}{\@ifnextchar.{\RulePremisesNewline[#1]{#2}}{\prooftree \justifies #2 \using{#1} \endprooftree}}}

\def \RuleWithInfo[#1]{\@ifnextchar[{\RuleWithNameAndCondition[#1]}{\RuleWithName[(#1)]}}

\def \RuleWithNameAndCondition[#1][#2]{\RuleWithName[(#1)^{#2}]}

\def \Inf[#1]{\proofrulebaseline=2ex \abovedisplayskip12\point\belowdisplayskip12\point \abovedisplayshortskip8\point\belowdisplayshortskip8\point \@ifnextchar[{\RuleWithInfo}{\RuleWithName[#1]}}

\usepackage{adjustbox}

\usepackage[misc,geometry]{ifsym}
\usepackage{orcidlink}

\begin{document}
\title{Sound Automation of Magic Wands
\\(extended version)
}
%
%

\author{Thibault Dardinier\inst{1}\orcidlink{0000-0003-2719-4856} \and
Gaurav Parthasarathy\inst{1} \and
Noé Weeks\inst{2} \and\\
Peter M{\"u}ller\inst{1}\orcidlink{0000-0001-7001-2566} \and
Alexander J.~Summers\inst{3}\orcidlink{0000-0001-5554-9381}}

\authorrunning{T. Dardinier et al.}
%
\institute{Department of Computer Science, ETH Zurich, Switzerland\\
\email{\{thibault.dardinier, gaurav.parthasarathy, peter.mueller\}@inf.ethz.ch}
\and
École Normale Supérieure, France\\
\email{noe.weeks@ens.psl.eu}\\
\and
University of British Columbia, Canada\\
\email{alex.summers@ubc.ca}
}

\maketitle              
\begin{abstract}
The magic wand $\wand$ (also called separating implication) is a separation logic connective commonly used to specify properties of partial data structures, for instance during iterative traversals. A \emph{footprint} of a magic wand formula $A \wand B$ is a state that, combined with any state in which $A$ holds, yields a state in which $B$ holds. The key challenge of proving a magic wand (also called \emph{packaging} a wand) is to find such a footprint. Existing package algorithms either have a high annotation overhead or, as we show in this paper, are unsound.

We present a formal framework that precisely characterises a wide design space of possible package algorithms applicable to a large class of separation logics. We prove in Isabelle/HOL that our formal framework is sound and complete, and use it to develop a novel package algorithm that offers competitive automation and is sound. Moreover, we present a novel, restricted definition of wands and prove in Isabelle/HOL that it is possible to soundly combine fractions of such wands, which is not the case for arbitrary wands. We have implemented our techniques for the Viper language, and demonstrate that they are effective in practice.

\end{abstract}

\section{Introduction}\label{sec:intro}
Separation logic~\cite{Reynolds02a} (SL hereafter) is a program logic that has been widely used to prove complex properties of heap-manipulating programs.
The two main logical connectives that enable such reasoning are the \emph{separating conjunction} $*$ and the \emph{separating implication} (more commonly known as the \emph{magic wand}) \wand, in combination with \emph{resource assertions} which represent \eg{} exclusive ownership of (and permission to access) particular heap locations.
The separating conjunction expresses that two assertions prescribe ownership of disjoint parts of the heap, useful, for instance, to reason about aliasing or race conditions.
More precisely, the assertion $A*B$ holds in a program state $\sigma$ if and only if $\sigma$ can be split into two \emph{compatible} program states $\sigma_A$ and $\sigma_B$ such that $A$ and $B$ hold in $\sigma_A$ and $\sigma_B$, respectively.
In SL, heaps of program states are \emph{partial} maps from locations to values; their domains represent heap locations exclusively owned. Two program states are compatible if (the domains of) their heaps are disjoint.

Intuitively, a magic wand $A \wand B$ can be used to express the difference between the heap locations that $B$ and $A$ provide permission to access.
The magic wand is useful, for instance, to specify partial data structures, where $B$ specifies the entire data structure and $A$ specifies a part that is missing~\cite{Tuerk10,Maeda11}.
$A \wand B$ holds in a state $\sigma_w$, if and only if for \emph{any} program state $\sigma_A$ in which $A$ holds and that is compatible with $\sigma_w$, $B$ holds in the state obtained by combining the heaps of $\sigma_A$ and $\sigma_w$.
Thus, if $A * (A \wand B)$ holds in a state, then so does $B$, analogously to the \emph{modus ponens} inference rule in propositional logic.

The magic wand has been shown to enable or greatly simplify proofs in many different cases~\cite{Yang01,Krishnaswami06,HaackHurlin09,Tuerk10,Maeda11,Dodds11,Cao2019,AstrauskasMuellerPoliSummers19b}.
For instance, Yang~\cite{Yang01} uses the magic wand to prove the Schorr-Waite graph marking algorithm.
Dodds~\etal~\cite{Dodds11} employ the wand to specify synchronisation barriers for deterministic parallelism.
Examples using magic wands to specify partial data structures include tracking ongoing traversals of a data structure~\cite{Tuerk10,Maeda11}, where the left-hand side of the wand specifies the part of the data structure yet to be traversed, or for specifying protocols that enforce orderly modification of data structures~\cite{Krishnaswami06,HaackHurlin09,Jensen11} (\eg{} the protocol governing Java iterators).
More recently, wands have been used for formal reasoning about borrowed references in the Rust programming language, which employs an ownership type system to ensure memory safety~\cite{AstrauskasMuellerPoliSummers19b}.
Magic wands \asout{can }concisely \as{represent}\asout{specify} the \as{\emph{remainder} of a} data structure from which a borrowed reference was taken, \as{as well as reflecting back modifications to the part}\asout{partial data structure} accessible via the reference.
\thibault{\asout{As an}\as{For} example, consider a struct \code{Point} (represented by a SL predicate \code{Point}) with two fields \code{x} and \code{y} of type \code{i32} (represented by the SL predicate \code{i32}).
A Rust method that takes as input a Point \code{p} and returns \as{a borrow of} its \asout{first coordinate}\as{field} \code{x} \as{is}\asout{could be} specified with the postcondition
$\code{int32(x)} * (\code{int32(x)} \wand \code{Point(p)})$,
thus enabling the caller to regain ownership of the entire data structure \code{Point(p)}.}

The complexity of SL proofs has given rise to a variety of automatic SL verifiers that reduce the required proof effort.
Given the usefulness of magic wands, it is important that such verifiers also provide automatic support for wands.
However, reasoning about a magic wand requires reasoning about \emph{all} states in which the left-hand side holds, which is challenging.
It has been shown that a separation logic even without the separating conjunction (but with the magic wand) is as expressive as a variant of second-order logic and, thus, undecidable~\cite{Brochenin12}.

Two different approaches~\cite{SchwerhoffSummers15,BlomHuisman15} that provide partially-automated support are implemented in the verifiers Viper~\cite{MuellerSchwerhoffSummers16} and VerCors~\cite{BlomDHO17}.
However, the approach implemented in VerCors~\cite{BlomHuisman15} incurs significant annotation overhead, and the approach  in Viper~\cite{SchwerhoffSummers15} suffers from a fundamental, previously undiscovered flaw that renders the approach unsound.
Both approaches require user-provided \emph{package operations} to direct the verifier's proof search.
\emph{Packaging} a wand $A \wand B$ expresses that the verifier should prove and subsequently record $A \wand B$.
To package $A \wand B$ the verifier must split the current state into two compatible states $\sigma'$ and $\sigma_w$ such that $A \wand B$ holds in $\sigma_w$.
We call $\sigma_w$ a \emph{footprint} of the wand.
After successfully packaging a wand, the verifier must disallow changes to $\sigma_w$ \tout{in order }to preserve the wand's validity: the verifier \emph{packages the footprint into the wand}.


The key challenge for supporting magic wands in automatic verifiers is to define a \emph{package algorithm} that packages a wand.
In VerCors's package algorithm\peter{~\cite{BlomHuisman15}}, a user must manually specify a footprint for the wand and the algorithm checks whether the wand holds in the specified footprint.
This leads to a lot of annotation overhead.
Viper's current package algorithm\peter{~\cite{SchwerhoffSummers15}} reduces this overhead significantly by \pout{defining an algorithm to} automatically infer\peter{ring} a suitable footprint.
Unfortunately, as we show in this paper, Viper's current algorithm has a fundamental flaw that causes the algorithm to infer an \emph{incorrect} footprint in certain cases\pout{, \ie{} a state in which the wand does \emph{not} hold.
As a result, one can prove results that do not hold}, which \peter{may lead} to unsound reasoning.
We will explain the fundamental flaw in~\secref{sec:problem}; it illustrates the subtlety of supporting this important connective.

\subsubsection{Approach and Contributions.}
In this paper, we present a formal foundation for sound package algorithms\pout{ (key to automating wands in automatic SL verifiers)}, and we implement a novel such
algorithm based on these foundations.
Our algorithm requires the same annotation overhead as the prior, flawed Viper algorithm, which is (to our knowledge) the most automatic existing approach.
%
We introduce a formal framework expressed via a novel \emph{package logic} that defines the design space for package algorithms.
The soundness of a package algorithm can be justified by showing that the algorithm finds a proof in our package logic.
The design space for package algorithms is large since there are various aspects that affect how one expresses the algorithm including (1)~which footprint an algorithm infers or checks \pout{for a wand} (there are often multiple options\peter{, see}\pout{as we will show in~}~\secref{sec:package_logic}), (2)~the state model (which differs between different SL verifiers),  and (3)~restricted definitions of wands (for instance, to ensure each wand has a unique minimal footprint).
Our package logic deals with (1) by capturing all sound
derivations for the same wand.
To deal with (2) and (3), our logic is parametric along multiple dimensions.
For instance, the state model can be any separation algebra to support different SL extensions (e.g.\ fractional permissions~\cite{Boyland}).

Our logic also supports parameters to restrict the allowed footprints for wands in systematic ways.
Such restrictions are useful, for instance, in a logic supporting \emph{fractional permissions}.
Fractional permissions permit splitting ownership/resources into shared fragments which typically permit read access to the underlying data.
However, as we show in~\secref{sec:combinable}, fractional parts of general magic wands cannot always be soundly recombined.
Existing solutions for other connectives impose side conditions to enable sound recombinations~\cite{LeHobor18}, which are often hard to check automatically.
We instead introduce a novel restriction of magic wands to avoid such side conditions and develop a corresponding second package algorithm again based on the formal framework provided by our package logic.
\noindent
We make the following contributions:
\begin{itemize}
    \item We formalise a \emph{package logic} that can be used as a basis for a wide range of package algorithms (\secref{sec:package_logic}).
    The logic has multiple parameters including: a separation algebra to model the states and a parameter to restrict the definition of a wand in a systematic way.
    We formally prove the logic sound and complete for any instantiation of the parameters in Isabelle/HOL.~\cite{PackageLogicAFP}

    \item We develop a novel, restricted definition of a wand  (\secref{sec:combinable}) and prove in Isabelle/HOL that this wand can always be recombined~\cite{CombinableWandsAFP}.

    \item
        We implement sound package algorithms for both the standard and the restricted wand in the Viper verifier and justify their soundness directly via our package logic (\secref{sec:automation}).
        We evaluate both algorithms on the Viper test suite. Our evaluation shows that (1) our algorithms perform similarly well to prior work and correctly reject examples where prior work is unsound, and (2) our restricted wand definition is expressive enough for most examples.
\end{itemize}
\thibault{Our Isabelle formalisation and the implementation of our new package algorithm are publicly available~\cite{PackageLogicAFP,CombinableWandsAFP,artifact}.}
\tout{We will make our Isabelle formalisation~\cite{PackageLogicAFP,CombinableWandsAFP} and \as{the implementation of our new} package algorithm\gaurav{s} available as part of our artefact~\cite{artifact}.}

\section{Background and Motivation}\label{sec:problem}
In this section, we present the necessary background for this paper.
We use \emph{implicit dynamic frames}~\cite{SmansEA09} to represent SL assertions, since both existing automatic verifiers that support wands (VerCors and Viper) are based on it.
There is a known strong correspondence between SL and implicit dynamic frames~\cite{ParkinsonSummers12}.
%

\subsection{Implicit Dynamic Frames}\label{subsec:idf_background}
Just like SL assertions, implicit dynamic frames (IDF hereafter) assertions specify not only value information, but also \emph{permissions} to heap locations that are allowed to be accessed.
To justify dereferencing a heap location, the corresponding permission is required, ensuring memory safety.
IDF assertions specify permissions to locations and value information separately.
An assertion $\code{acc(x.val)}$ (an \emph{accessibility predicate}) denotes permission to the heap \emph{location} \code{x.val}, while $\code{x.val} = v$ expresses that \code{x.val} contains \emph{value} $v$. The separating conjunction in IDF enforces disjointness (formally: acts multiplicatively) with respect to resource assertions such as accessibility predicates; in particular,
if $\code{acc(x.val)} * \code{acc(y.val)}$ holds in a state, then \code{x} and \code{y} must be different (analogously to SL).

The main difference between IDF and SL is that SL does not allow general heap-dependent expressions such as \code{x.val = v} or \code{x.left.right}~\cite{SmansEA09} to be specified separately from the permissions to the heap locations they depend on. The IDF assertion $\code{acc(x.val)} * \code{x.val = v}$ must be expressed in SL via the \emph{points-to assertion} \pointsto{\code{x.val}}{\code{v}}, which also conveys exclusive permission to the location \code{x.val}. IDF supports heap dependent expressions
within \emph{self-framing} assertions: those which require permissions to all the heap locations on whose values they depend (e.g.\ $\code{acc(x.val)} * \code{x.val = v}$ is self-framing but \code{x.val = v} is not)~\cite{SmansEA09}.

\subsection{A Typical Example Using Magic Wands}\label{subsec:tree_example}
\begin{figure}[t]
\begin{minipage}[t]{0.60\textwidth}
\begin{viper2linenum}
method leftLeaf(x: Ref) : (y: Ref)
  requires Tree(x)
  ensures Tree(x) {
  y := x
  @\ghostopcode{package}@ Tree(x) $\wand$ Tree(x) @\label{line:leftLeaf_package_init}@

  while(y.left != null)
    inv Tree(y) $*$ (Tree(y) $\wand$ Tree(x)) {
    y := y.left
    @\ghostopcode{package}@ Tree(y) $\wand$ Tree(x) @\label{line:leftLeaf_package_step}@
      // { hints for package} @\label{line:leftLeaf_proof_script}@
  }
  @\ghostopcode{apply}@ Tree(y) $\wand$ Tree(x) @\label{line:leftLeaf_apply_post}@
}
\end{viper2linenum}
\end{minipage}
\hfill
\begin{minipage}[t]{0.40\textwidth}
\begin{viper2}
Tree(x: Ref) $\triangleq$
 acc(x.val) *
 acc(x.left) * acc(x.right)
 (x.left != null $\Rightarrow$
            Tree(x.left)) *
 (x.right != null $\Rightarrow$
            Tree(x.right))
\end{viper2}
\end{minipage}
\caption{The code on the left \as{finds}\asout{computes} the leftmost leaf of a binary tree and includes specifications to prove memory safety. The predicate describing the permissions of a tree is defined on the right. The loop invariant uses a wand to summarise the permissions of the input tree excluding the tree \as{not yet}\asout{yet to be} traversed. The blue operations are ghost operations to guide the verifier\as{; we}\asout{We}
omit \as{those}\asout{ghost operations} specific to predicates.
The \code{package} \pout{(lines~\ref{line:leftLeaf_package_step} and~\ref{line:leftLeaf_proof_script})} requires further hints in existing approaches\peter{, see}\pout{(the hints vary based on the approach) as discussed in}~\happref{J}.
\pout{Here for instance, \thibault{to package the wand}, one must specify
that\tout{ in order to package the wand,} the wand \as{from} the loop invariant must be applied.}
}
\label{fig:leftLeaf}
\end{figure}
\figref{fig:leftLeaf} shows a variation of an example from the VerifyThis \pout{verification} competition~\cite{HuismanKM15}.
The method \code{leftLeaf} \peter{iteratively} computes the leftmost leaf of a binary tree\pout{ in an iterative fashion} (\code{package} and \code{apply} operations, shown in blue, should be ignored for now).
The pre- and postconditions of \code{leftLeaf} are both \code{Tree(x)}, which is a \emph{predicate instance} used to specify all permissions to \peter{the fields of}\pout{\code{val}, \code{left}, and \code{right} fields in} the tree rooted at \code{x} (the recursive definition of this predicate is on the right of~\figref{fig:leftLeaf}).
Proving this specification amounts to proving that \code{leftLeaf} is memory-safe and that the permissions to the input tree are preserved,
\peter{enabling further calls}\pout{which is typically necessary to enable further such methods to be called} on the same tree.

The key challenge when verifying \code{leftLeaf} is specifying an appropriate loop invariant.
The loop invariant must track the permissions to the subtree rooted at \code{y} that still needs to be traversed, since otherwise dereferencing \code{y.left} in the loop body is not allowed.
Additionally, the invariant must track all of the remaining permissions in the input tree rooted at \code{x} (the permissions to the nodes already traversed and others unreachable from \code{y}), since otherwise the postcondition cannot be satisfied.
The former can be easily expressed with \code{Tree(y)}.
The latter can be elegantly achieved with a magic wand $\code{Tree(y)} \wand \code{Tree(x)}$.
This wand promises $\code{Tree(x)}$ if one combines the wand with $\code{Tree(y)}$.
That is, the wand represents (at least) the difference between the permissions making up the two trees.
Using SL's modus-ponens-like inference rule (directed by the \code{apply} operation on line~\ref{line:leftLeaf_apply_post}, explained next), one can show that the loop invariant entails the postcondition.

\subsection{Wand Ghost Operations}\label{subsec:wand_annot}
Automatic SL verifiers such as GRASShopper~\cite{PiskacWZ14}, VeriFast~\cite{JacobsSPVPP11}, VerCors, and Viper generally represent permissions owned by a program state in two ways: by recording predicate instances (such as \code{Tree(x)} in~\figref{fig:leftLeaf}) and \emph{direct} permissions to heap locations.
Magic wand instances provide a third way to represent permissions and are recorded analogously. Verifiers that support them require two wand-specific \emph{ghost operations}, which instruct the verifiers when to prove a wand and when to apply a recorded wand instance using SL's modus-ponens-like rule.

A \code{package} ghost operation expresses that a verifier should prove a new wand instance in the current state and report an error if the proof attempt fails.
To prove a new wand instance, the verifier must split the current state into two states $\sigma'$ and $\sigma_w$ such that the wand holds in the \emph{footprint} state $\sigma_w$; on success, permissions in the footprint are effectively exchanged for the resulting magic wand instance.
We call a procedure that selects a footprint by splitting the current state a \emph{package algorithm}.
On lines~\ref{line:leftLeaf_package_init} and~\ref{line:leftLeaf_package_step} of~\figref{fig:leftLeaf}, new wands are packaged to establish and preserve the invariant, respectively.

The \code{apply} operation \emph{applies} a wand $A \wand B$ using SL's modus-ponens-like rule if the verifier records a wand instance of $A \wand B$ and $A$ holds in the current state (and otherwise fails), exchanging these for the assertion $B$.
The \code{apply} operation is directly justified by the wand's semantics: Combining a wand's footprint with \emph{any} state in which $A$ holds is guaranteed to yield a state in which $B$ holds.
For the \code{apply} operation on line~\ref{line:leftLeaf_apply_post} of~\figref{fig:leftLeaf}, the verifier removes the applied wand instance and \code{Tree(y)}, in exchange for the predicate instance \code{Tree(x)}.

\subsection{The Footprint Inference Attempt (FIA)}\label{subsec:footprint_infer_base}
Package algorithms differ in how a footprint for the specified magic wand is selected.
In \peter{VerCors}\pout{one existing approach}~\cite{BlomHuisman15}, the user must manually provide the footprint and the algorithm checks whether the specified footprint is correct.
In \peter{Viper's current} \pout{the other} approach~\cite{SchwerhoffSummers15}, a footprint is inferred. We explain and compare to the latter approach since it is the more automatic of the two; hereafter, we refer to its package algorithm as \emph{the Footprint Inference Attempt} (\emph{FIA}).
Inferring a correct footprint is challenging due the complexity of the wand connective.
In particular, we have discovered that, in certain cases, the FIA infers \emph{incorrect} footprints, leading to unsound reasoning\footnote{\peter{This unsoundness might not be observable in restricted logics, but it is in Viper (see \happref{B}) and  the rich  logics supported by existing verification tools.}}.
The goal of this subsection is to understand the FIA's key ideas, which our solution will build on, and why it is unsound.

In general, there may be multiple valid footprints for a magic wand $A \wand B$.
The FIA attempts to infer a footprint which is as close as possible to the \emph{difference} between the permissions required by $B$ and $A$, taking as few permissions as possible while aiming for a footprint compatible with $A$ (so that the resulting wand can be later applied)~\cite{SchwerhoffSummers15}.
That is, the FIA includes only permissions in the footprint it infers that are specified by $B$ \emph{and not} guaranteed by $A$.

For a wand $A \wand B$, the FIA constructs an arbitrary state $\sigma_A$ that satisfies $A$ (representing $\sigma_A$ symbolically).
Then, the FIA tries to construct a state $\sigma_B$ in which $B$ holds by taking permissions (and copying corresponding heap values) from $\sigma_A$ if possible and the current state otherwise.
If this algorithm succeeds, the (implicit) inferred footprint consists of the permissions that were taken from the current state.
The FIA constructs $\sigma_B$ by iterating over the permissions and logical constraints in $B$.
For each permission, the FIA checks whether $\sigma_A$ owns the permission.
If so, the FIA adds the permission to $\sigma_B$ and removes the permission from $\sigma_A$.
Otherwise, the FIA removes the permission from the current state or fails if the current state does not have the permission.
For each logical constraint, the FIA checks that the constraint holds in $\sigma_B$ as constructed so far.
We show an example of the FIA correctly packaging a wand in \happref{A}.

\subsubsection{Unsoundness of the FIA.}
\label{subsec:unsoundness}
We have discovered that for some wands $A \wand B$, the FIA  determines an \emph{incorrect} footprint for the magic wand. This unsoundness can arise when the FIA performs a case split on the content of the arbitrary state $\sigma_A$ satisfying $A$.
In such situations, the FIA infers a footprint for each case \emph{separately}, making use of properties that hold in that case. For certain wands, this leads to different footprints being selected for each case, while \emph{none} of the inferred footprints can be used to justify $B$ in \emph{all} cases, \ie{} for \emph{all} states $\sigma_A$ that satisfy $A$. As a result, the packaged wand does \emph{not} hold in any of the inferred footprints, which can make verification unsound, as we illustrate below.

The wand $w := \code{acc(x.f)} * (\code{x.f = y} \vee  \code{x.f = z}) \wand \code{acc(x.f)} * \code{acc(x.f.g)}$ illustrates the problem.
For this wand, every state $\sigma_A$ satisfying the left-hand side must have permission to \code{x.f}.
However \code{x.f} may either point to \code{y} or \code{z}.
If \code{x.f} points to \code{y} in $\sigma_A$, then to justify the right-hand side's second conjunct, the footprint must contain permission to \code{y.g}.
Analogously, if \code{x.f} points to \code{z} in $\sigma_A$, then the footprint must contain permission to \code{z.g}.
The wand's semantics requires a footprint to justify the wand's right-hand side for all states in which the left-hand side holds, and thus, a correct footprint must be able to justify \emph{both} cases.
Hence, the footprint must have permission to \emph{both} \code{y.g} and \code{z.g}.
However, the FIA's inferred footprint is in effect the disjunction of these two permissions.

Packaging the above wand $w$ using the FIA leads to unsound reasoning.
After the incorrect package described above in a state with permission to \code{x.f}, \code{y.g}, and \code{z.g}, the assertion $\code{acc(x.f)} * (\code{acc(y.g)} \vee \code{acc(z.g)}) * w$ can be proved since the FIA removes permission to either \code{y.g} or \code{z.g} from the current state, but not both. However, this assertion does not actually hold! According to the semantics of wands, $w$'s footprint must include permission to \code{x.f} or permission to both \code{y.g} and \code{z.g},
which implies that the assertion $\code{acc(x.f)} * (\code{acc(y.g)} \vee \code{acc(z.g)}) * w$
is equivalent to false.
\pout{As we show in \appref{app:viper-unsound}, this unsoundness can be observed in Viper, since it implements the FIA.}

The unsoundness of the FIA shows the subtlety and challenge of developing sound package algorithms. Algorithms that soundly infer a single footprint for all states in which the wand's left-hand side holds must be more involved than the FIA\@.
Ensuring their soundness requires a \emph{formal} framework to construct them and justify their correctness. We introduce such a framework in the next section.

\section{A Logical Framework for Packaging Wands}\label{sec:package_logic}
In this section, we present a new logical framework that defines the design space for (sound) package algorithms. The core of this framework is our \emph{package logic}, which defines the space of potential algorithmic choices of a footprint for a particular magic wand. Successfully packaging a wand in a given state is (as we will show) equivalent to finding a derivation in our package logic, and any actual package algorithm must correspond to a proof search in our logic (if it is sound). In particular, we provide soundness (\thmref{thm:soundness}) and completeness (\thmref{thm:completeness}) results for our logic. We define a specific package algorithm with this logic at its foundation, inspired by the FIA package algorithm~\cite{SchwerhoffSummers15}
(described in \secref{subsec:footprint_infer_base}) but amending its unsoundness, resulting in (to the best of our knowledge) the first sound and relatively automatic package algorithm.

All definitions and results in this section have been fully mechanised\thibault{~\cite{PackageLogicAFP}} in Isabelle/HOL. Our mechanised definitions are parametric with the underlying verification logic in various senses: the underlying separation algebra is a parameter, the syntax of assertions is defined in a way which allows simple extension with different base cases and connectives, and the semantics of magic wands itself can be restricted if only particular kinds of footprint are desired in practice. As a specific example of the latter parameter, in \secref{sec:combinable} we define a novel restriction of magic wand footprints which guarantees better properties in combination with certain usages of fractional permissions; this is seamlessly supported by the general package logic presented here. Nonetheless, to simplify the exposition of this section, we will assume that any magic wand footprint satisfying the connective's standard semantics is an acceptable result.

\subsection{Footprint Selection Strategies}\label{subsec:motivation}
As we explained in \secref{sec:intro}, there is a wide design space for package algorithms; in particular, many potential strategies for finding a magic wand's footprint exist and none is clearly optimal.
Recall that a footprint is a state, and thus consists of permissions to certain heap locations as well as storing their corresponding values; for simplicity we identify a footprint by the permissions it contains.

For example, consider the following magic wand (using fractional permissions) \\
$\code{acc(x.b, 1/2)}\; \wand\; \code{acc(x.b, 1/2)} * (\code{x.b} \Rightarrow \code{acc(x.f)})$.
Suppose this magic wand is to be packaged in a state where full permissions to both \code{x.b} and \code{x.f} are held, and the value of \code{x.b} is currently false. Two valid potential footprints are:
\begin{enumerate}
\item Full permission to \code{x.f}. This is sufficient to guarantee the right-hand side will hold regardless of the value that \code{x.b} has by the time the wand is applied.
\item Half permission to \code{x.b}. By including this permission, the fact that \code{x.b} is currently false is also included, and thus permission to \code{x.f} is not needed.
\end{enumerate}
There is no clear reason to prefer one choice over the other: different package algorithms (or manual choices) might choose either. Our package logic allows either choice along with any of many less optimal choices, such as taking both permissions.
On the other hand, as motivated earlier in \secref{subsec:motivation}, our package logic must (and does) enforce that a single valid footprint is chosen for a wand that works for each and every potential state satisfying its left-hand side.

\subsection{Package Logic: Preliminaries}
\gout{In order}\gaurav{To} capture different state models and \gout{different}flavours of separation logic, our package logic is parameterised by a separation algebra.
For space reasons, we present here a simplified overview of this \gout{separation}algebra, but all definitions \thibault{(\gout{along with the semantics of our}\gaurav{including our} assertion \gout{language}\gaurav{semantics})} are given in \happref{D} and have been mechanised\gout{in Isabelle/HOL}.
We consider a separation algebra~\cite{Calcagno2007,Dockins2009} where $\Sigma$ is the set of states,
$\oplus: \Sigma \times \Sigma \rightarrow \Sigma$ is a partial operation that is commutative and associative,
and $e \in \Sigma$, which corresponds to the empty state, is a neutral element for $\oplus$.
We write $\succeq$ for the induced partial order of the resulting partial commutative monoid, and
$\sigma_1 \# \sigma_2$ iff $\sigma_1 \oplus \sigma_2$ is defined (\ie{} $\sigma_1$ and $\sigma_2$ are compatible).
Finally, if $\sigma_2 \succeq \sigma_1$, we define the subtraction $\sigma_2 \ominus \sigma_1$ to be the $\succeq$-largest state $\sigma_r$ such that
$\sigma_2 = \sigma_1 \oplus \sigma_r$.

We define our package logic for an assertion language with the following grammar:
$A = A{*}A \mid \mathcal{B} {\Rightarrow} A \mid \mathcal{B}$,
where $A$ ranges over assertions and $\mathcal{B}$ over \emph{semantic assertions}.
\tout{The semantics of this assertion language is formally defined in \appref{app:context}.}
To allow our package logic to be applied to a variety of underlying assertion logics, we distinguish only the two most-relevant connectives:
the separating conjunction and an implication (for expressing conditional assertions). To support additional \gout{connectives and base-cases}\gaurav{constructs} of the assertion logic, the  third type of assertion we consider is a \gout{general kind of}\emph{semantic assertion}, \ie{}~\gaurav{a} function from $\Sigma$ to Booleans. This third type can be instantiated to represent logical assertions that \gout{don't}\gaurav{do not} match the first two cases.
In particular, assertions such as \code{x.f = 5}, \code{acc(x.f)}, abstract predicates (such as \code{Tree(x)}) or magic wands can be represented as semantic assertions.
This core assertion language can also be easily extended with native support for \eg{} the logical conjunction and disjunction connectives; we explain in \happref{E} how to extend the rules of the logic accordingly.

\subsection{The Package Logic}\label{subsec:logic}
We define our package logic to prescribe the design space of algorithms for deciding how, in an initial state $\sigma_0$, to select a valid footprint (or fail) for a magic wand $A \wand B$.
The aim is to infer states $\sigma_w$ and $\sigma_1$ that partition $\sigma_0$ (\ie{}~$\sigma_0 = \sigma_1 \oplus \sigma_w$) such that $\sigma_w$ is a valid footprint for $A \wand B$ (when combined with any compatible state satisfying $A$, the resulting state satisfies $B$). In particular, all permissions (and logical facts) required by the assertion $B$ must either come from the footprint or be guaranteed to be provided by any compatible state satisfying $A$.

Recall from \secref{subsec:footprint_infer_base} that the mistake underlying the FIA approach ultimately resulted from allowing multiple different footprints to be selected conditionally on a state satisfying $A$, rather than a single footprint which works for all such states. Our package logic addresses this concern by defining judgements in terms of the \emph{set} of all states satisfying $A$; whenever \emph{any} of these tracked states is insufficient to provide a permission required by $B$, our logic will force this permission to be added \emph{in general} to the wand's footprint (taken from the current state).

A \emph{witness set} $S$ is a set of pairs of states $(\sigma_A, \sigma_B)$;
 conceptually, the first represents the state available for trying to prove $B$ \emph{in addition} to the current state; this is initially a state satisfying the wand's left-hand side $A$. The second represents the state assembled (so-far) to attempt to satisfy the right-hand side $B$. We write $S^1$ for the set of first elements of all pairs in a witness set $S$. A \emph{context} $\Delta$ is a pair $(\sigma,S)$ of a state and a witness set; here, $\sigma$ represents the (as-yet unused remainder of the) current state in which the wand is being packaged.

The basic idea behind a derivation in our logic is to show how to assemble a witness set in which \emph{all} second elements are states satisfying $B$, via some combinations of: (1)~moving a part of the first element of a pair in the witness set into the second, and (2)~moving a part of the outer state $\sigma$ into \emph{all} first elements of the pairs (this becomes a part of the wand's footprint). The actual judgements of the logic are a little more complex, to correctly record any hypotheses (called \emph{path-conditions}) that result from deconstructing conditional assertions in $B$.

\subsubsection{Configurations and Reductions.}
A \emph{configuration} represents a current objective in our package logic:
the part of the wand's right-hand side still to be satisfied as well as the current state of a footprint computation. A configuration is a triple $\langle B, \pc, (\sigma, S) \rangle$, where $B$ is an assertion, $pc$ is a \emph{path condition} (a function from $\Sigma$ to Booleans), and $(\sigma, S)$ is a context. Conceptually, $B$ is the assertion still to be satisfied, \pc{} represents hypotheses we are currently working under, and the context $(\sigma, S)$ tracks the current state and witness set, as described above.

A \emph{reduction} is a judgement $\langle B, \pc, (\sigma_0, S_0) \rangle \rightsquigarrow (\sigma_1, S_1)$, representing the achievement of the objective described via the configuration on the left, resulting in the final context on the right; $\sigma_1$ is the new version of the outer state (and becomes the new current state after the \code{package} operation); whatever was removed from the initial outer state is implicitly the selected footprint state $\sigma_w$. If a reduction is derivable in our package logic, this footprint $\sigma_w$ guarantees that for all $(\sigma_A, \sigma_B) \in S_0$,
if $(\sigma_A \oplus \sigma_B) \# \sigma_w$, then $\sigma_A\oplus\sigma_w$ satisfies $\pc \Rightarrow B$.
The condition $(\sigma_A \oplus \sigma_B) \# \sigma_w$ ensures that the pair $(\sigma_A, \sigma_B)$ actually corresponds to a state in which the wand can be applied given the chosen footprint $\sigma_w$, as we explain later.
The package logic defines the steps an algorithm may take to achieve this goal.

We represent packaging a wand $A \wand B$ in state $\sigma_0$ by the derivation of a reduction
$\langle B, \lambda \sigma \ldotp \top, (\sigma_0, \{ (\sigma_A, e) \mid \sigma_A \models A \} ) \rangle \rightsquigarrow (\sigma_1, S_1)$,
for some state $\sigma_1$ and witness set $S_1$.
The path condition is initially true (we are not yet under any \pout{particular} hypotheses). The initial witness set contains all pairs of a state $\sigma_A$ that satisfies $A$ and the empty state $e$, to which a successful reduction will add permissions in order to satisfy $B$\footnote{If $B$ is intuitionistic, this can be simplified to only the $\succeq$-\emph{minimal} states that satisfy $A$. $B$ is intuitionistic~\cite{Reynolds02a} iff, if $B$ holds in a state $\sigma$, then $B$ holds in any state $\sigma'$ such that $\sigma' \succeq \sigma$.
In intuitionistic SL or in IDF, all assertions are intuitionistic.}.
\pout{Note that a}
An actual algorithm \pout{based on our package logic} need not explicitly compute this \peter{(possibly infinite)} set\pout{(which could well be infinite)}, but can instead track it symbolically.
If the algorithm finds a derivation of this reduction, it has proven that the difference between $\sigma_0$ and $\sigma_1$ is a valid footprint of the wand $A \wand B$,
since the logic is sound (\thmref{thm:soundness} below).

\subsubsection{Rules.}
\begin{figure}[t]
\begin{center}
    \[
    \begin{array}{c}

    \Inf[\mathit{Implication}]{\langle A,
    \lambda \sigma \ldotp pc(\sigma) \land b(\sigma), \Delta \rangle \rightsquigarrow \Delta'}
        {\langle b \Rightarrow A, pc, \Delta \rangle \rightsquigarrow \Delta'}

    \hspace{10mm}

    \Inf[\mathit{Star}]{\langle A_1, pc, \Delta_0 \rangle \rightsquigarrow \Delta_1}.
        {\langle A_2, pc, \Delta_1 \rangle \rightsquigarrow \Delta_2}
        {\langle A_1 * A_2, pc, \Delta_0 \rangle \rightsquigarrow \Delta_2}

    \\[3em]

    \Inf[\mathit{Atom}]{\forall (\sigma_A, \sigma_B) \in S \ldotp pc(\sigma_A) \Longrightarrow \sigma_A \succeq \mathit{choice}(\sigma_A, \sigma_B) \land \mathcal{B}(\mathit{choice}(\sigma_A, \sigma_B))}.
        { S_{\top} = \{ (\sigma_A \ominus \mathit{choice}(\sigma_A, \sigma_B) , \sigma_B \oplus \mathit{choice}(\sigma_A, \sigma_B) ) | (\sigma_A, \sigma_B) \in S \land pc(\sigma_A) \} }.
        { S_{\bot} = \{ (\sigma_A, \sigma_B) | (\sigma_A, \sigma_B) \in S \land \lnot pc(\sigma_A) \} }
        {\langle \mathcal{B}, pc, (\sigma, S) \rangle \rightsquigarrow (\sigma, S_{\top} \cup S_{\bot})}

    \\[3em]

    \Inf[\mathit{Extract}]
    {\sigma_0 = \sigma_1 \oplus \sigma_w}{\textsf{stable}(\sigma_w)}{\langle A, pc, (\sigma_1, S_1) \rangle \rightsquigarrow \Delta}.
    {S_1 = \{ (\sigma_A \oplus \sigma_w, \sigma_B) | (\sigma_A, \sigma_B) \in S_0 \land (\sigma_A \oplus \sigma_B) \# \sigma_w \} }
        {\langle A, pc, (\sigma_0, S_0) \rangle \rightsquigarrow \Delta}

    \end{array}
    \]
    \end{center}
    \caption{Rules of the package logic.}
    \label{fig:package_rules}
\end{figure}
\figref{fig:package_rules} presents the four rules of our logic, defining (via derivable reductions) how a configuration can be reduced to a context. There is a rule for each type of assertion $B$: \emph{Implication} for an implication, \emph{Star} for a separating conjunction, and \emph{Atom} for a semantic assertion. The logic also includes the rule \emph{Extract}, which represents a choice to extract permissions from the outer state and adds them to all pairs of states in the witness set.
In the following, we informally write \emph{reducing an assertion}
to refer to the process of deriving (in the logic) that the relevant configuration containing this assertion reduces to some context.

To reduce an implication $\mathcal{B} \Rightarrow A$, the rule \textit{Implication} conjoins the hypothesis $\mathcal{B}$ with the previous path condition, leaving $A$ to be reduced. Informally, this expresses that satisfying $pc \Rightarrow (b \Rightarrow A)$ is equivalent to satisfying $(pc \land b) \Rightarrow A$.

For a separating conjunction $A_1 * A_2$, the \emph{Star} rule expresses that both $A_1$ and $A_2$ must be reduced, in order to reduce $A_1 * A_2$; permissions used in the reduction of the first conjunct must not be used again, which is reflected by the threading-through of the intermediate context $\Delta_1$.\footnote{The order in the premises is unimportant since $A_1 * A_2$ and $A_2 * A_1$ are equivalent.}


The \emph{Atom} rule specifies how to prove that all states in $S^1$ (where $S$ is the witness set) satisfy the assertion $pc \Rightarrow \mathcal{B}$.
To understand the premises, consider a pair $(\sigma_A, \sigma_B) \in S$.
If $\sigma_A$ does not satisfy the path condition, \ie{} $\lnot pc(\sigma_A)$, then $\sigma_A$ \emph{does not} have to justify $\mathcal{B}$,
and thus the pair $(\sigma_A, \sigma_B)$ is left unchanged; this case corresponds to the set $S_\bot$.
Conversely, if $\sigma_A$ satisfies the path condition, \ie{} $pc(\sigma_A)$, then $\sigma_A$ must satisfy $\mathcal{B}$,
and the corresponding permissions must be transferred from $\sigma_A$ to $\sigma_B$.
Since some assertions may be satisfied in different ways, such as disjunctions, the algorithm has a choice in how
to satisfy $\mathcal{B}$, which might be different for each pair $(\sigma_A, \sigma_B)$.
This choice is represented by $\mathit{choice}(\sigma_A, \sigma_B)$,
which must satisfy $\mathcal{B}$ and be smaller or equal to $\sigma_A$.
We update the witness set by transferring $\mathit{choice}(\sigma_A, \sigma_B)$ from $\sigma_A$ to $\sigma_B$.
This second case corresponds to the set $S_\top$.
Note that the \emph{Atom} rule can be applied only if $\sigma_A$ satisfies $\mathcal{B}$, for all pairs $(\sigma_A, \sigma_B) \in S$ such that $pc(\sigma_A)$.
If not, a package algorithm must either first extract more permissions from the outer state with the \emph{Extract} rule, or fail.

The $\mathit{Extract}$ rule (applicable at any step of a derivation), expresses that we can extract permissions
(the state\footnote{We explain formally in \happref{D} the notion of a stable state, which is a technicality of our general state model; in standard SL\tout{ and many variants}, all states are stable.}
 $\sigma_w$) from the outer state $\sigma_0$, and combine them with the first element of each pair of states in the witness set.
Note that $(\sigma_A, \sigma_B)$ is removed from the witness set if $\sigma_A \oplus \sigma_B$ is not compatible with $\sigma_w$.
In such cases, adding $\sigma_w$ to $\sigma_A$ would create a pair in the witness set representing a state in which the wand cannot be applied. Consequently, there is no need to establish the right-hand side of the wand for this pair and our logic correspondingly removes it.
Finally, the rule requires that we reduce the assertion $A$ in the new context.

\gout{The}\gaurav{A} \gout{strategy of a}package algorithm\gaurav{'s strategy} is mostly reflected by \gout{when and}how it uses the \emph{Extract} rule.
To package \gout{the wand}$\code{acc(x.b, 1/2)} \wand \code{acc(x.b, 1/2)} * (\code{x.b} \Rightarrow \code{acc(x.f)})$ from \secref{subsec:motivation}
one algorithm might use this rule to extract permission to \code{x.f}; another might use it to extract
permission to \code{x.b} (if \code{x.b} had value false in the original state).
\tout{\appref{app:example-derivation} shows a full derivation of a reduction in our logic.}

\subsubsection{Example of a Derivation}
\label{app:example-derivation}

\newcommand{\scode}[1]{\code{\footnotesize{}#1}}

\thibault{Let us now} illustrate how these rules can be used to package the wand from \secref{subsec:motivation},
$w := \code{acc(x.f)} * (\code{x.f = y} \vee \code{x.f = z}) \wand \code{acc(x.f)} * \code{acc(x.f.g)}$.
We omit the path condition since it is always the trivial condition ($\lambda \sigma \ldotp \top)$.
Assume that the outer state $\sigma_0$ is the addition of $\sigma_{yz}$, a state that contains permission to \code{y.g} and \code{z.g}, and $\sigma_1$.
$S_0 := \{ (\sigma_A, e) \mid \sigma_A \in \Sigma \land \sigma_A \models \code{acc(x.f)} * (\code{x.f = y} \vee \code{x.f = z}) \}$ is the initial witness set.
We show below a part of a proof that
$\langle \scode{acc(x.f) * acc(x.f.g)}, (\sigma_0, S_0) \rangle \rightsquigarrow (\sigma_1, S_3)$ is correct,
and thus that $\sigma_{yz}$ is a correct footprint of the wand $w$
(since $\sigma_0 = \sigma_1 \oplus \sigma_{yz}$):

\vspace{5mm}
\begin{adjustbox}{max width=\textwidth}
\Inf[Star]{
    \Inf[Atom]
    { \ldots }
    { \langle \scode{acc(x.f)}, (\sigma_0, S_0) \rangle \rightsquigarrow (\sigma_0, S_1) }
}
{
    \Inf[Extract]{
        \Inf[Atom]
        { \ldots }
        { \langle \scode{acc(x.f.g)}, (\sigma_1, S_2) \rangle \rightsquigarrow (\sigma_1, S_3) }
    }{\dagger}
    {\langle \scode{acc(x.f.g)}, (\sigma_0, S_1) \rangle \rightsquigarrow (\sigma_1, S_3)}
}
{\langle \scode{acc(x.f) * acc(x.f.g)}, (\sigma_0, S_0) \rangle \rightsquigarrow (\sigma_1, S_3) }
\end{adjustbox}

\vspace{5mm}


\gout{The}\gaurav{This} derivation\thibault{, which \gout{corresponds to}\gaurav{reflects} the \gaurav{package algorithm} \gout{strategy}\gaurav{that we will describe}\gout{described} in \secref{subsec:fixing},} can be read from bottom to top and from left to right.
Using the rule $\mathit{Star}$, we \gout{first}split the assertion\gout{\code{acc(x.f) * acc(x.f.g)}} into its two conjuncts,
\code{acc(x.f)} (on the left) and \code{acc(x.f.g)} (on the right).
We then handle \code{acc(x.f)} using the rule $\mathit{Atom}$.
\code{acc(x.f)} holds in the first element of each pair of $S_0$, since any state that satisfies the \gaurav{wand's} left-hand side\gout{of the wand} \gout{has permission to}\gaurav{owns} \code{x.f}.
Therefore, we use the rule $\mathit{Atom}$ with a $\mathit{choice}$ function that always chooses the relevant state with exactly full permission to \code{x.f}.
$S_1$ is the updated witness set where this permission to \code{x.f} has been transferred from the first to the second element of each pair of states.
Next, we \gout{want to}handle \code{acc(x.f.g)}\gout{using the rule $\mathit{Atom}$}.
\gaurav{We cannot do this directly using the rule $\mathit{Atom}$ from \gout{the witness set}$S_1$.}
\gout{However, we cannot do this directly from the witness set $S_1$:}
We know that, for each $(\sigma_A, \sigma_B) \in S_1$,
\code{x.f.g} evaluated in $\sigma_A$ is either \code{y} or \code{z},
but $\sigma_A$ \gout{does not have any permission to}\gaurav{owns neither} \code{y.g} \gout{or}\gaurav{nor} \code{z.g}.
\gout{Thus}\gaurav{So}, we transfer the permissions to both \code{y.g} and \code{z.g} from the outer state $\sigma_0$ to all states of $S_1^1$,
using the rule $\mathit{Extract}$, which results in the \gout{new}context $(\sigma_1, S_2)$\thibault{; $\dagger$ represents the three other premises of the rule, namely $\sigma_0 = \sigma_{yz} \oplus \sigma_1$, $\textsf{stable}(\sigma_{yz})$, and \gout{the definition of}$S_2$\gaurav{'s definition}.}
Finally, we apply the rule $\mathit{Atom}$ to prove
${ \langle \code{acc(x.f.g)}, (\sigma_1, S_2) \rangle \rightsquigarrow (\sigma_1, S_3) }$,
where the $\mathit{choice}$ function chooses for each pair the corresponding state that contains full permission to \code{x.f.g}.

\tout{This example follows the general proof search strategy,
which deconstructs the right-hand side using the rules $\mathit{Star}$ and $\mathit{Implication}$.
When we arrive at a semantic assertion, we first try to prove it directly from the witness set (using the rule $\mathit{Atom}$), as we did to prove \code{acc(x.f)}.
If we cannot prove this semantic assertion directly, we use the rule $\mathit{Extract}$ to extract the necessary resources from the outer state before using the rule $\mathit{Atom}$,
as we did to prove \code{acc(x.f.g)}.
This general proof search strategy corresponds to the package algorithm described in \secref{subsec:fixing}.
Different heuristics in when and how to use the rule \emph{Extract} lead to different proof strategies, and thus to different package algorithms, as explained in \secref{subsec:motivation}.
Given that our logic is sound, any package algorithm that corresponds to a proof search strategy in the package logic is sound.}

\subsection{Soundness and Completeness}\label{subsec:soundness-completeness}

We write $\vdash \langle B, pc, \Delta \rangle \rightsquigarrow \Delta'$ to express that a reduction can be derived in the logic.
As explained above, the goal of a package algorithm is to find a derivation of
$\langle B, \lambda \_ \ldotp \top, (\sigma, \{ (\sigma_A, e) \mid \sigma_A \in S_A \}) \rangle \rightsquigarrow (\sigma', S')$.
If it succeeds, then the difference between $\sigma'$ and $\sigma$ is a valid footprint of $A \wand B$,
since our package logic is sound.
In particular, we have proven the following soundness result in Isabelle/HOL:


\begin{theorem}\textbf{Soundness}. \label{thm:soundness}
    Let $B$ be a well-formed\protect\footnote{\label{note:wf}We formally define well-formedness in \happref{D}.
        Intuitively, a well-formed assertion roughly corresponds to a self-framing assertion as defined in \secref{subsec:idf_background}.
        } assertion. If
        \begin{enumerate}
            \item the set $S_A$ contains all states that satisfy $A$. \ie{} $\forall \sigma_A \ldotp \sigma_A \models A \Rightarrow \sigma_A \in S_A$,
            \item $\vdash \langle B, \lambda \_ \ldotp \top, (\sigma, \{ (\sigma_A, e) \mid \sigma_A \in S_A \}) \rangle \rightsquigarrow (\sigma', S')$, and
            \item at least one of the following conditions holds:
            \begin{enumerate}
                \item $B$ is intuitionistic
                \item For all $(\sigma_A, \sigma_B) \in S'$, $\sigma_A$ contains no permission (\ie{} $\sigma_A \oplus \sigma_A = \sigma_A$)
            \end{enumerate}
        \end{enumerate}
    then there exists a stable state $\sigma_w$ s.t.
    $\sigma = \sigma' \oplus \sigma_w$ and
    $\sigma_w$ is a footprint of $A \wand B$.
\end{theorem}

The third premise shows that, in an intuitionistic SL or in IDF,
the correspondence between a derivation in the logic and a valid footprint of a wand is straightforward (case (a)).
However, in classical SL, one must \tout{in general }additionally check that all permissions in the witness set have been consumed (case (b)).

We have also proved in Isabelle/HOL that our package logic is complete, \ie{} \emph{any} valid footprint
can be computed via a derivation in our package logic:

\begin{theorem}\textbf{Completeness}.
    \label{thm:completeness}
    Let $B$ be a well-formed\footnotemark[\getrefnumber{note:wf}] 
assertion.
    If
        $\sigma_w$ is a stable footprint of $A \wand B$, and
        $\sigma = \sigma' \oplus \sigma_w$,
    then there exists a witness set $S'$ such that
    $\vdash \langle B, \lambda \_ \ldotp \top, (\sigma, \{ (\sigma_A, e) \mid \sigma_A \in S_A \}) \rangle \rightsquigarrow (\sigma', S')$.
\end{theorem}


\subsection{A Sound Package Algorithm}\label{subsec:fixing}

We now describe an automatic package algorithm that corresponds to a proof search strategy in our package logic, and which is thus sound.
To convey the main ideas, consider packaging a wand of the shape $A \wand B_1 * \ldots * B_n$.\footnote{In \happref{I}, we \gout{additionally}\gaurav{also} show how our package algorithm handles implications.}
Our algorithm traverses the assertion $B_1 * \ldots * B_n$ from left to right, similarly to the FIA approach; this traversal is justified by repeated applications of the rule \textit{Star}.
Assume at some point during this traversal that the current context is $(\sigma_0, S)$.
When we encounter the assertion $B_i$, we have two possible cases:
\begin{enumerate}
    \item All states $\sigma_A \in S^1$ satisfy $B_i$,
    which means that the permissions (or values) required by $B_i$ are provided by the left-hand side of the wand.
    In this case, for each pair $(\sigma_A, \sigma_B) \in S$, we transfer permissions (and the corresponding values) to satisfy $B_i$ from $\sigma_A$ to $\sigma_B$, using the rule \textit{Atom}.
    Note that the transferred permissions might be different for each pair $(\sigma_A, \sigma_B)$.
    This gives us a new witness set $S'$, while the outer state $\sigma_0$ is left unchanged.
    We must then handle the next assertion $B_{i+1}$ in the context $(\sigma_0, S')$.
    \item There is at least one pair $(\sigma_A, \sigma_B) \in S$ such that $B_i$ does not hold in $\sigma_A$.
    In this case, the algorithm fails if combining the permissions (and values) contained in the outer state
    with each $\sigma_A \in S^1$ is not sufficient to satisfy $B_i$.
    Otherwise, we apply the rule \textit{Extract} to transfer permissions from the outer state $\sigma_0$
    to each state $\sigma_A$ in $S^1$ such that $B_i$ holds in $\sigma_A$.
    This gives us a new context $(\sigma_0', S')$.
    We can now apply the first case with the context $(\sigma_0', S')$.
\end{enumerate}

\section{Using the Logic with Combinable Wands}\label{sec:combinable}
Extending SL with fractional permissions~\cite{Boyland} is well-known to be useful for reasoning about heap-manipulating concurrent programs with shared state.
In this setting, permission amounts are generalised to fractions $0 \le p \le 1$.
Reading a heap location is permitted if $p > 0$, and writing if $p = 1$, which permits concurrent reads and ensures exclusive writes. The assertion
\code{acc(x.f, $p$)} holds in a state that has \emph{at least} $p$ permission to \code{x.f}.
A permission amount $p + q$ to a heap location \code{x.f} can be split into a permission amount $p$ and a permission amount $q$,
i.e.\ $\code{acc(x.f, p + q)} \models \code{acc(x.f, p) * acc(x.f, q)}$,
and these two permissions can be recombined, i.e.\ $\code{acc(x.f, p) * acc(x.f, q)} \models \code{acc(x.f, p + q)}$.

This concept has been generalised~\cite{Boyland10,JacobsPiessens11,Dinsdale-YoungDGPV10,LeHobor18,Brotherston20} to \emph{fractional assertions} $A^p$, representing a fraction $p$ of $A$.
$A^p$ holds in a state $\sigma$ iff there exists a state $\sigma_A$ in which $A$ holds and $\sigma$ is obtained from $\sigma_A$ by multiplying all permission amounts held by $p$
\cite{LeHobor18,Brotherston20}; in this case, we write
$\sigma = p \cdot \sigma_A$.
For example, $\code{acc(x.f)}^p \equiv \code{acc(x.f, p)}$, and $\code{Tree(x)}^p$ (where \code{Tree} is the predicate defined in \figref{fig:leftLeaf})
expresses $p$ permission to all nodes of the tree rooted in \code{x}.

Using fractional assertions, one might specify a function \code{find},
which searches a binary tree and yields a subtree whose root contains key \code{key}, as follows~\cite{Brotherston20}:\\
$
\{ \; \code{Tree(x)}^p \; \} \; \code{find(x, key)} \; \{ \; \lambda \code{ret}. \left( \code{Tree(ret)} * (\code{Tree(ret)} \wand \code{Tree(x)} ) \right)^p \; \}
$,
in which \code{ret} corresponds to the return value of \code{find}.
This postcondition is similar to the loop invariant in \figref{fig:leftLeaf}, except that it needs only a fraction $p$ of \code{Tree(x)}.
A number of automatic SL verifiers, such as Caper~\cite{Dinsdale-YoungP17}, Chalice~\cite{LeinoMuellerSmans09}, VerCors~\cite{BlomDHO17}, VeriFast~\cite{JacobsSPVPP11}, and Viper~\cite{MuellerSchwerhoffSummers16},
support fractional assertions in some form.


\subsubsection{Combinable Assertions.}
While it is always possible to split an assertion $A^{p+q}$ into $A^p * A^q$, recombining $A^p * A^q$ into $A^{p+q}$ is sound only under some conditions,
for example~\cite{LeHobor18} if $A$ is \emph{precise} (in the usual SL sense~\cite{Reynolds02a}).
We say that $A$ is \emph{combinable} iff the entailment $A^p * A^q \models A^{p+q}$ holds for any two positive fractions $p$ and $q$ such that $p + q \le 1$.
As an example, \code{acc(x.f)} is combinable, but $\code{acc(x.f)} \vee \code{acc(x.g)}$ is not because
a state containing half permission to both \code{x.f} and \code{x.g} satisfies
$(\code{acc(x.f)} \vee \code{acc(x.g)})^{0.5} * (\code{acc(x.f)} \vee \code{acc(x.g)})^{0.5}$,
but not $\code{acc(x.f)} \vee \code{acc(x.g)}$.
Combinable assertions are particularly useful to reason about concurrent programs, for instance, to combine the postconditions of parallel branches when they terminate~\cite{Brotherston20}.

However, a magic wand is in general \emph{not} combinable, as we show below.
This is problematic for SL verifiers; they cannot soundly combine wands, nor predicates that could possibly contain wands in their bodies. One way to prevent the latter is to forbid magic wands in predicate bodies entirely,
but this limits the common usage of predicates to abstract over general assertions in specifications~\cite{ParkinsonBierman05}.
Another solution is to disallow combining fractional instances of a predicate if its body contains a wand, which means requiring additional annotations to ``taint'' such predicates transitively.
This is overly restrictive for wands which are actually combinable and complicates reasoning about abstract predicate families~\cite{ParkinsonBierman05}.

To address this issue, we propose a novel restriction of the wand, called \emph{combinable wand} (we use \emph{standard wand} to refer to the usual, unrestricted connective).
Unlike standard wands in general, a combinable wand is always combinable if its right-hand side is combinable.
Thus, by only using combinable wands instead of standard wands, all assertions in logics such as those employed by VerCors and Viper can be made combinable without any of the other
aforementioned restrictions regarding predicates.
\secref{sec:evaluation} shows that the restriction combinable wands impose is sufficiently weak for practical purposes.
Finally, footprints of combinable wands can be automatically inferred by package algorithms built on our package logic.
All results in this section have been proven in Isabelle/HOL.

\subsubsection{Standard Wands are not Combinable in General.}\label{subsec:explanation}

Even if $B$ is combinable, the standard wand $A \wand B$ is, in general, not. As an example, the wand $w := \code{acc(x.f, 1/2)} \wand \code{acc(x.g)}$ is not combinable,
because $w^{0.5} * w^{0.5} \not\models w$.
To see this, consider two states $\sigma_f$ and $\sigma_g$, containing full permissions to only \code{x.f} and \code{x.g}, respectively.
Both states are valid footprints of $w$, \ie{} $\sigma_f \models w$ (because $\sigma_f$ is incompatible with all states that satisfy the left-hand side)
and $\sigma_g \models w$ (because $\sigma_g$ entails the right-hand side)\@.
Thus, by definition, $0.5 \cdot \sigma_f \models w^{0.5}$ and $0.5 \cdot \sigma_g \models w^{0.5}$.
However, $0.5 \cdot \sigma_f \oplus 0.5 \cdot \sigma_g$, i.e.\ a state with half permission to both \code{x.f} and \code{x.g}, is \emph{not} a valid footprint of $w$, and thus $w^{0.5} * w^{0.5} \not\models w$.

Intuitively, $w$ is not combinable because one of its footprints, $\sigma_f$, is incompatible with the left-hand side of the wand,
but becomes compatible when the footprint is
scaled down to a fraction.\tfootnote{Maybe not needed: In
this example, $\sigma_f$ is incompatible with \emph{all} states that satisfy the left-hand side.
However, footprints that are incompatible with only \emph{some} states that satisfy the left-hand side might also
render wands not combinable, as we show in \appref{app:example-not-combinable}.}
After scaling, the wand no longer holds trivially, and the footprint does not necessarily establish the right-hand side\@.

To make this intuition more precise, we introduce the notion of \emph{scalable footprints}.
For a state $\sigma$, we define $\mathit{scaled}(\sigma)$ to be the set of copies of $\sigma$ multiplied by any fraction $0 < \alpha \le 1$,
\ie{} $\mathit{scaled}(\sigma) := \{ \alpha \cdot \sigma \mid 0 < \alpha \le 1 \}$.
A footprint $\sigma_w$ is \emph{scalable \wrt{} a state} $\sigma_A$ iff
either (1) $\sigma_A$ is compatible with \emph{all} states from $\mathit{scaled}(\sigma_w)$,
or (2) $\sigma_A$ is compatible with \emph{no} state in $\mathit{scaled}(\sigma_w)$.
A footprint is \emph{scalable for a wand} $A \wand B$ iff it is scalable w.r.t. all states that satisfy A.
Intuitively, this means that the footprint does not ``jump'' between satisfying the wand trivially and having to satisfy the right-hand side\@.
In the above example, $\sigma_g$ is a scalable footprint for $w$, but $\sigma_f$ is not.

\subsubsection{Making Wands Combinable.}\label{subsec:solution-combinable}

The previous paragraphs show that, even if $B$ is combinable,
the standard wand $A \wand B$ is in general not combinable because it can be satisfied by non-scalable footprints.
Therefore, we define a novel restricted interpretation for wands that \emph{forces} footprints to be scalable, in the following sense.
The restricted interpretation of a wand accepts all scalable footprints, and transforms non-scalable footprints before checking whether they actually satisfy the wand.
We call a wand with this restricted interpretation a \emph{combinable wand},
and write $A \cwand B$ to differentiate it from the standard wand $A \wand B$.

For standard wands, \emph{any} state $\sigma_w$ is a footprint of $A \wand B$ iff,
for all states $\sigma_A$ that satisfy $A$, $\sigma_A \# \sigma_w \Rightarrow \sigma_A \oplus \sigma_w \models B$.
We obtain the definition of combinable wands by replacing $\sigma_w$ with a (possibly smaller) state $\mathcal{R}(\sigma_A, \sigma_w)$ that is scalable w.r.t. $\sigma_A$.
$\mathcal{R}(\sigma_A, \sigma_w)$ is defined as $\sigma_w$ if \emph{no} state in $\mathit{scaled}(\sigma_w)$ is compatible with any $\sigma_A$;
in that case, condition (2) of scalable footprints holds for $\mathcal{R}(\sigma_A, \sigma_w)$ w.r.t. $\sigma_A$.
Otherwise, $\mathcal{R}(\sigma_A, \sigma_w)$ is obtained by removing just enough permissions from  $\sigma_w$ to ensure that \emph{all} states in $\mathit{scaled}(\mathcal{R}(\sigma_A, \sigma_w))$ are compatible with $\sigma_A$,
which ensures that condition (1) holds for $\mathcal{R}(\sigma_A, \sigma_w)$ w.r.t. $\sigma_A$.

To formally define $\mathcal{R}(\sigma_A, \sigma_w)$, we
fix a concrete separation algebra (formally defined in \happref{G}),
whose states are pairs $(\pi, h)$ of a \emph{permission mask} $\pi$, which maps heap locations to fractional permissions,
and a partial heap $h$, which maps heap locations to values.

\begin{definition}\label{def:restricted-footprint}
    Let $(\pi_A, h_A)$ and $(\pi_w, h_w)$ be two states,
    and let $\pi'_w$ be the permission mask such that $\forall l \ldotp \pi_w'(l) = \min(\pi_w(l), 1 - \pi_A(l))$. Then
    $$
    \mathcal{R}((\pi_A, h_A), (\pi_w, h_w)) =
    \begin{cases}
        (\pi_w, h_w) & \text{if } \forall \sigma \in \mathit{scaled}((\pi_w, h_w)) \ldotp \lnot (\pi_A, h_A)\# \sigma \\
        (\pi_w' , h_w) & \text{otherwise}
    \end{cases}
    $$

\noindent
    The \emph{combinable wand} $A \cwand B$ is then interpreted as follows:
    $$
    \sigma_w \models A \cwand B \Longleftrightarrow
    \left( \forall \sigma_A \ldotp \sigma_A \models A \land \sigma_A \# \mathcal{R}(\sigma_A, \sigma_w)
    \Longrightarrow
    \sigma_A \oplus \mathcal{R}(\sigma_A, \sigma_w)
    \models B \right)
    $$
\end{definition}

The following theorem (proved in Isabelle/HOL) shows some key properties of combinable wands.

\begin{theorem}\label{thm:combinable}
    Let $B$ be an intuitionistic assertion.
    \begin{enumerate}
        \item If $B$ is combinable, then $A \cwand B$ is combinable.
        \item $A \cwand B \models A \wand B$.
        \item If $A$ is a binary assertion, then $A \cwand B$ and $A \wand B$ are equivalent.
    \end{enumerate}
\end{theorem}

Property~1 expresses that combinable wands constructed from combinable assertions are combinable, which enables verification methodologies underlying
tools such as VerCors and Viper to support flexible combinations of wands and predicates (as motivated at the start of this section). Property~2 implies that $A * (A \cwand B) \models B$, that is, combinable wands can be applied like standard wands. Property~3 states that combinable wands pose no restrictions if the left-hand side is binary, that is, if it can be expressed without fractional permissions (formally defined in \happref{G}). For example, the predicate \code{Tree(x)} from \figref{fig:leftLeaf} is binary, which implies that the wands $\code{Tree(y)} \cwand \code{Tree(x)}$ and $\code{Tree(y)} \wand \code{Tree(x)}$ are equivalent. This property is an important reason for why combinable wands are expressive enough for practical purposes, as we further evidence in \secref{sec:evaluation}.

Footprints of combinable wands can be automatically inferred by \gout{package}algorithms built on our package logic.
We explain \thibault{(along with \gout{some}examples)} in \happref{H} how to lift the package logic presented in \secref{sec:package_logic}
to handle alternative definitions of allowable \gout{wand}footprints such as the restrictions imposed by Def.~\ref{def:restricted-footprint}.

\section{Evaluation}\label{sec:automation}\label{sec:evaluation}
We have implemented package algorithms for the standard wands and combinable wands in a custom branch of Viper's~\cite{MuellerSchwerhoffSummers16} verification condition generator (VCG).
Both are based on the package logic described in \secref{sec:package_logic}, adapted to the fractional permission setting.
Both algorithms automate the proof search strategy outlined in \secref{subsec:fixing}.
Viper's VCG translates Viper programs to Boogie~\cite{LeinoBoogie2} programs.
It uses a total-heap semantics of IDF~\cite{ParkinsonSummers12}, where Viper states include a heap and a permission mask (tracking  fractional permission amounts).
The heap and mask are represented in Boogie as maps; we also represent witness sets as Boogie maps.

We evaluate our implementations of the package algorithms on Viper's test suite and compare them to Viper's implementation of the FIA as presented in \secref{subsec:footprint_infer_base}. Our key findings are that our algorithms
(1)~enable the verification of almost all correct \code{package} operations.
(2)~correctly report \code{package} operations that are supposed to fail (in contrast to the FIA), and
(3)~have an acceptable performance overhead compared to the FIA\@.
Moreover, interpreting wands as combinable wands as explained in \secref{subsec:solution-combinable} has only a minor effect on the results, but correctly rejects attempts to package a non-combinable wand. This finding suggests that verifiers could improve their expressiveness by allowing flexible combinations of wands and predicates with only a minor completeness penalty.

For our evaluation, we considered all 85 files in the test suite for Viper's VCG \tout{that contain}\thibault{with} at least one \code{package} operation. From these 85 files, we removed 29 files containing features that our implementation does not yet support. 28 of these 29 files require proof scripts to guide the footprint inference, which are orthogonal to the concerns of this paper (see~\happref{J} for details).

\begin{table}[t]
\center
\begin{tabular}{lp{4mm}cp{4mm}cp{4mm}c}
 Algorithm && Expected result && Incorrectly verified && Spurious errors\\
\hline
FIA && 55 && 1 && 0 \\
S-Alg      && 51 && 0 && 5 \\
C-Alg      && 48 && 0 && 8 \\[2mm]
\end{tabular}
\caption{Verification results on our 56 benchmarks with the FIA, our algorithm for standard wands (S-Alg), and for combinable wands (C-Alg). For each algorithm, we report the number of correct verification results, false negatives, and false positives.}
\label{tbl:eval_classification}
\end{table}

Table~\ref{tbl:eval_classification} gives an overview of our results. These confirm that our algorithms for standard and combinable wands (S-Alg and C-Alg) do not produce false negatives, that is, are sound. In contrast, the FIA does verify an incorrect program (which is similar to the example in~\secref{subsec:footprint_infer_base}). While this is only a single unsound example, it is worth emphasing that (a) it comes from the pre-existing test suite of the tool itself, (b) the unsoundness was not known of until our work, and (c) soundness issues in a program verifier are critical to address; we show how to achieve this.

Compared with the FIA, our implementation reports a handful of false positives (spurious errors). For S-Alg, 3 out of 5 false positives are caused by missing features of our implementation (such as remembering a subset of the permissions that are inside predicate instances when manipulating predicates); these features could be straightforwardly added in the future.
The other 2 false positives are caused by S-Alg's strategy. In one, the only potential footprint prevents the wand from ever being applied; although technically a false positive, it seems useful to reject the wand and alert the user.
The other case is due to a coarse-grained heuristic applied by S-Alg that can be improved.

C-Alg reports the expected result in 48 benchmarks. Importantly, it correctly rejects one wand that indeed does not hold as a combinable wand. 5 of the 8 false positives are identical to those for S-Alg. In the other three benchmarks, the wands still do hold as combinable wands, but further extensions to C-Alg are required to handle them due to technical challenges regarding predicate instances. Once these extensions have been implemented, C-Alg will be as precise as S-Alg,  indicating that comparable program verifiers could switch to combinable wands to simply enable sound, flexible combinations with predicates.

To evaluate performance, we ran each of the three implementations 5 times on each of the 56 benchmarks on a Lenovo T480 with 32 GB of RAM and a i7-8550U 1.8 GhZ CPU, running on Windows 10.
We removed the slowest and fastest time, and then took the mean of the remaining 3 runs.
The FIA takes between 1 and 11 seconds per benchmark.
On average, S-Alg is 21\% slower than the FIA\@.
For 46 of the 56 examples, the increase is less than 30\%, and for 3 examples S-Alg is between a factor 2 and 3.4 slower.
The overhead is most likely due to the increased complexity of our algorithms, which track more states explicitly and require more quantified axioms in the Boogie encoding.
C-Alg is on average 10\% slower than S-Alg.
%
We consider the performance overhead of our algorithms to be acceptable, especially since wands occur much more frequently in our benchmarks than in average Viper projects, as judged by existing tests and examples. More representative projects will, thus, incur a much smaller slow-down.

\section{Related Work}\label{sec:related}
VerCors~\cite{BlomDHO17} and Viper~\cite{MuellerSchwerhoffSummers16} are to the best of our knowledge the only automatic SL verifiers that support magic wands.
Both employ \code{package} and \code{apply} ghost operations.
VerCors' package algorithm requires a user to manually specify a footprint whereas Viper infers footprints using the FIA, which is unsound as we show in~\secref{subsec:footprint_infer_base}.
Our package algorithm is as automatic as the FIA but is sound.

Lee and Park~\cite{LeePark14} develop a sound and complete proof system for SL including the magic wand.
Moreover, they derive a decision procedure from their completeness proof for propositional SL\@.
However, more expressive versions of SL (that include e.g.\ predicates and quantifiers) are undecidable~\cite{Brochenin12} and so this decision procedure cannot be directly applied in the logics employed by program verifiers.

\gaurav{
Chang~\etal~\cite{ChangRN07} define a shape analysis that derives magic wands $A \wand B$\as{ of a restricted form ($A$ and $B$ cannot contain general imprecise assertions); }
our package logic does not impose such restrictions\as{, which rule out some useful kinds of wands.}
%
For example, $A$ may be a data structure with a read-only part expressed via existentially-quantified fractional permissions or $A$ may contain the necessary permission to invoke a method, which may be an arbitrary assertion.
In follow-up work, Chang and Rival~\cite{ChangR08} present a restricted ``inductive'' magic wand.
Footprints of inductive \tout{magic }wands are expressed via a finite unrolling of \gout{the}\gaurav{an} inductive predicate defining $B$ until the permissions in $A$ are revealed.
Such wands are useful to reason about data structures with back-pointers such as doubly-linked lists}.

Iris~\cite{iris-ground-up} provides a custom proof mode~\cite{MoSelIris2018} for interactive SL proofs in Coq~\cite{coq}.
Separation logics expressed in Iris support wands and are more expressive than those of automatic SL verifiers at the cost of requiring more user guidance.
Packaging a wand in the proof mode requires manually specifying a footprint and proving that the footprint is correct.
While tactics can be used in principle to automate parts of this process, there are no specific tactics to infer footprints.

Fractional assertions have been used in various forms~\cite{Boyland10,JacobsPiessens11,Dinsdale-YoungDGPV10,LeHobor18,Brotherston20}.
Le and Hobor~\cite{LeHobor18} allow combining two fractional assertions $A^p$ and $A^q$ only if $A$ is \emph{precise} in the SL sense (\ie{} $A$ describes the contents of the heaps in which it holds precisely).
To avoid \gout{imposing a side condition}requiring $A$ to be precise,
Brotherston~\etal~\cite{Brotherston20} introduce \emph{nominal labels} \gout{that can be associated with}\gaurav{for} assertions.
If an assertion is split into two fractional assertions, then the same fresh label can be associated with both \gout{fractional}parts to indicate that they were split from the same assertion.
\gout{
Brotherston~\etal~allow combining two fractional assertions if both assertions are associated with the same label.}
\gaurav{Two fractional assertions with the same label can be combined.}
However, \gout{their}\gaurav{this} solution has not been implemented\gout{in a verifier} and \gout{their work}does not deal with packaging wands.
Our solution also avoids requiring that an assertion is precise and allows combining assertions even if they were not split from the same assertion. Instead of introducing \gout{nominal}labels, we introduce a light restriction that ensures that wands are always combinable.
As a result, \gout{general}assertions containing combinable wands but no other potentially imprecise connectives (such as disjunction) are combinable. In particular, all assertions employed in verifiers such as VerCors and Viper can be made combinable thanks to our work.

\section{Conclusion}\label{sec:conclusion}
We presented a package logic that precisely characterises sound package algorithms for automated reasoning about magic wands. Based on this logic, we developed a novel package algorithm that is inspired by an existing approach, but is sound. Moreover, we identified a sufficient criterion for wands to be combinable, such that they can be used flexibly in logics with fractional permissions, and presented a package algorithm for combinable wands. We implemented our solutions in Viper and demonstrated their practical usefulness. The soundness and completeness of our package logic, as well as key properties of combinable wands are all proved in Isabelle/HOL\@.
As future work, we plan to extend the implementation of the two package algorithms described in~\secref{sec:automation} by porting various features of the pre-existing FIA implementation. Moreover, we will use our package logic to develop another algorithm for Viper's symbolic-execution verifier.

\subsubsection*{Acknowledgements.}
This work was partially funded by the Swiss National Science Foundation (SNSF) under Grant No. 197065.

%
%
\bibliographystyle{splncs04}
\bibliography{references}

\newpage
\appendix
\section{Example for the Footprint Inference Attempt}\label{app:fia-example}

\begin{figure}[t]
  \centering
 \includegraphics[scale=0.55]{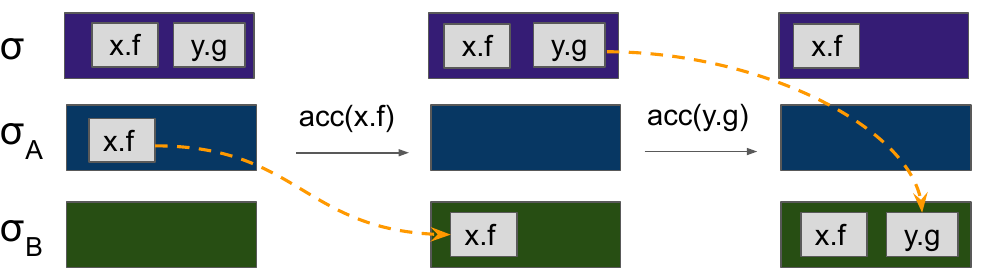}
 \caption{A visualisation of the footprint inference baseline traversing the right-hand side of $\code{acc(x.f)} \wand \code{acc(x.f)} * \code{acc(y.g)}$ in a state $\sigma$.
 $A$ holds in $\sigma_A$. Each small grey box represents a permission to a heap location.}
 \label{fig:fib_simple}
\end{figure}

\figref{fig:fib_simple} visualises the FIA for the right-hand-side of the wand $\code{acc(x.f)} \wand \code{acc(x.f)} * \code{acc(y.g)}$ in a current state $\sigma$ that has permissions to \code{x.f} and \code{y.g}.
$\sigma_A$ is an arbitrary state in which the wand's left-hand-side holds, and thus has permission to \code{x.f}.
When traversing the first conjunct \code{acc(x.f)}, the FIA removes permission to \code{x.f} from $\sigma_A$ and adds it to $\sigma_B$.
For the second conjunct \code{acc(y.g)}, the FIA removes permission to \code{y.g} from $\sigma$, since $\sigma_A$ does not contain any permission to \code{y.g}.
Since these were all the requirements from the wand's right-hand side, the FIA succeeds and the footprint (the permissions taken from $\sigma$) is a state with permission to \code{y.g}, which is a correct footprint.

\section{Unsoundness of the FIA in Viper}
\label{app:viper-unsound}

As we explained in \secref{subsec:unsoundness}, packaging the wand
$w := \code{acc(x.f)} * (\code{x.f = y} \vee  \code{x.f = z}) \wand \code{acc(x.f)} * \code{acc(x.f.g)}$
using the FIA leads to unsound reasoning:
Starting in a state with permission to \code{x.f}, \code{y.g}, and \code{z.g},
we can prove the assertion $\code{acc(x.f)} * (\code{acc(y.g)} \vee \code{acc(z.g)}) * w$.
However, a correct footprint of $w$ must either have some permission to \code{x.f}, or permission to \emph{both} \code{y.g} and \code{z.g}.
Therefore, $\code{acc(x.f)} * (\code{acc(y.g)} \vee \code{acc(z.g)}) * w$ is actually equivalent to false.

\begin{figure}
\begin{viper2}[numbers=left, stepnumber=1,firstnumber=1]
field f: Ref
field g: Int

method main(x:Ref, y:Ref, z:Ref)
    requires acc(x.f) && acc(y.g) && acc(z.g)
{
    package acc(x.f) && (x.f == y || x.f == z) --* acc(x.f) && acc(x.f.g)
    {
        assert x.f == y ? acc(y.g) : acc(z.g)
    }
    assert (acc(x.f) && (x.f == y || x.f == z) --* acc(x.f) && acc(x.f.g))
           && acc(x.f) && (perm(y.g) == write || perm(z.g) == write)
    if (perm(y.g) == write) {
        x.f := y
    }
    else {
        x.f := z
    }
    apply acc(x.f) && (x.f == y || x.f == z) --* acc(x.f) && acc(x.f.g)
    assert false
} 
\end{viper2}
\caption{A small Viper program that illustrates how to prove false using the unsoundness of the FIA.
This program relies on Viper's \emph{permission introspection} feature, which allows to inspect the amount of permission to a heap location owned by the current execution:
The expression \vipercode{perm(y.g)} yields the permission amount of \code{y.g} held by the current execution, not counting resources inside packaged wands.
}

\label{fig:viper-false}

\end{figure}


Viper currently implements the FIA, and it is possible to exploit the unsoundness of the FIA when packaging the wand $w$ to prove false.
While Viper does not directly support disjunctions of accessibility predicates, we can observe in \figref{fig:viper-false} that the assertion $\code{acc(x.f)} * (\code{acc(y.g)} \vee \code{acc(z.g)}) * w$
holds after packaging the wand $w$.
This example relies on Viper's \emph{permission introspection} feature:
The expression \vipercode{perm(y.g)} (for a reference \code{y} and a field \code{g}) yields the permission amount of \code{y.g} held by the current execution, not counting resources inside packaged wands.

The program shown in \figref{fig:viper-false} is currently verified by Viper.
Method \code{main} starts in a state with permission to \code{x.f}, \code{y.g}, and \code{z.g}.
We then package the wand $w$ (lines 7-10) using the FIA, and help it a bit with a \emph{proof script} (line 9, see \appref{app:proofscript} for more details).
After the package, we assert
$w * \code{acc(x.f)} * (\code{acc(y.g)} \vee \code{acc(z.g)})$ (lines 11-12), using permission introspection to express the disjunction.\footnote{Contrary to accessibility predicates (such as \code{acc(y.g)}), Viper allows combining disjunctions with permission introspection.}
We can actually prove false explicitly using this magic wand.
To do this, we assign \code{y} to \code{x.f} if the current execution owns \code{y.g}, and \code{z} otherwise (lines 13-18), using permission introspection.
Finally, we apply the wand (line 19).

Viper is able to prove false on line 20, because:
\begin{itemize}
    \item Either the current execution owns \code{y.g}, in which case the permission of \code{z.g} was computed as the footprint of $w$.
    Thus, the current execution satisfies the assertion $w * \code{acc(x.f)} * \code{acc(y.g)}$.
    Applying the wand $w$ with \code{x.f = y} effectively exchanges ownership of \code{x.f} with ownership of \code{y.g},
    resulting in a state that owns \code{y.g} twice, which is thus an inconsistent state.
    \item Or the current execution does not own \code{y.g}, which means that the permission of \code{y.g} was computed as the footprint of $w$,
        and thus the execution satisfies the assertion $w * \code{acc(x.f)} * \code{acc(z.g)}$.
        In this case, assigning \code{z} to \code{x.f} and then applying the wand $w$ leads to an inconsistent state, which owns \code{z.g} twice.
\end{itemize}

\section{Example for a Package Algorithm Based on the Package Logic}\label{app:example-sound-package}

\begin{figure}[t]
    \includegraphics[width=\textwidth]{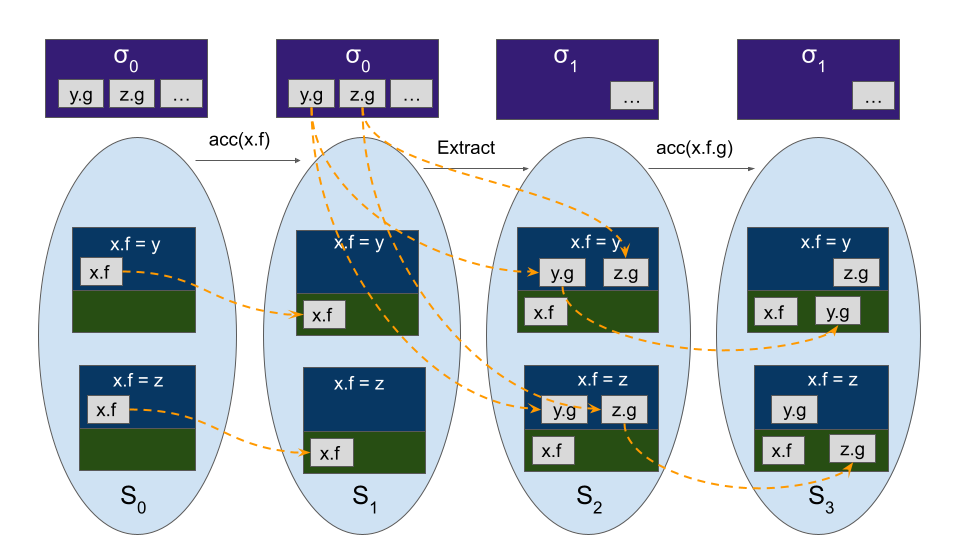}
    \caption{Illustration of how the algorithm described in \secref{subsec:fixing} packages the wand
    $\code{acc(x.f)} * (\code{x.f = y} \vee  \code{x.f = z}) \wand \code{acc(x.f)} * \code{acc(x.f.g)}$.
    $\sigma_0$ represents the program state before the package, and $\sigma_1$ is the outer state
    after permissions to \code{y.g} and \code{z.g} have been extracted (which correspond to a footprint of the wand inferred by the algorithm).
    The $S_i$ represent witness sets, i.e. sets of pairs of states.
    Pairs of states are represented as stacks of two states.
    Finally, permissions to heap locations are represented with light grey rectangles.
    }
    \label{fig:new_package}
\end{figure}

We illustrate in \figref{fig:new_package} the algorithm described in \secref{subsec:fixing}
Our goal is to package the wand
$\code{acc(x.f)} * (\code{x.f = y} \vee  \code{x.f = z}) \wand \code{acc(x.f)} * \code{acc(x.f.g)}$.
Recall that the FIA infers an incorrect footprint for this wand, as explained in \secref{subsec:footprint_infer_base}.
We assume that the initial outer state $\sigma_0$ contains permissions to $y.g$, $z.g$, and some other resources.
The initial context is $(\sigma_0, S_0)$, where
$S_0^1$ contains two minimal states that satisfy $\code{acc(x.f)} * (\code{x.f = y} \vee  \code{x.f = z})$ (the left-hand side of the wand).
The two states in $S_0^1$ have permission to \code{x.f}, and \code{x.f} contains value \code{y} in one state and \code{z} in the other one.
The second element of both pairs in $S_0$ is the empty state.

We first handle the first conjunct of the right-hand side, \code{acc(x.f)}.
Since both states in $S_0^1$ satisfy \code{acc(x.f)}, the first case applies.
For the two pairs $(\sigma_A, \sigma_B) \in S_0$, we transfer permission to \code{x.f} from $\sigma_A$ to $\sigma_B$, and we obtain the new witness set $S_1$.
We then handle the second conjunct of the right-hand side, \code{acc(x.f.g)}.
Since no state $\sigma_A \in S_1^1$ satisfies \code{acc(x.f.g)} (since the two states in $S_1^1$ have no permissions anymore), the second case applies.
We need to transfer permissions from the outer state $\sigma_0$ to all states of $S_1^1$.
Moreover, since \code{x.f} evaluates to \code{y} in one state of $S_1^1$ and to \code{z} in the other one,
we transfer permissions to both \code{y.g} and \code{z.g} to the states of $S_1^1$,
and we obtain the new context $(\sigma_1, S_2)$ (where $\sigma_1$ is $\sigma_0$ without permissions to \code{y.g} and \code{z.g}).
For each pair of states $(\sigma_A, \sigma_B) \in S_2$, we can now transfer permissions from $\sigma_A$ to $\sigma_B$ to satisfy \code{acc(x.f.g)}.

In the end, the footprint inferred is the difference between the final outer state $\sigma_1$ and the initial outer state $\sigma_0$,
i.e. the state that contains permissions to \thibault{(and values of)} exactly \code{y.g} and \code{z.g}.
Therefore, the algorithm presented above finds the correct footprint for this wand, contrary to the FIA as we explained in \secref{subsec:footprint_infer_base}.
After the package the new state of the program is $\sigma_1$, in which we record the wand instance and proceed with the verification of the subsequent statements.

\section{Separation Algebra and Assertions}\label{app:context}

In this section, we formally define the separation algebra and the assertion language that our package logic builds on.

\begin{definition}\label{def:separation-algebra}
    A \emph{separation algebra} is a quintuple $(\Sigma, \oplus, e, |\_|, \textsf{stable})$, where:
    \begin{enumerate}
        \item $\Sigma$ is a set of states, $\oplus$ is a partial addition on $\Sigma$ that is commutative and associative,
        and $e$ is the neutral element of $\oplus$.
        \item $|\_|$ (the \emph{core}) is a function from $\Sigma$ to $\Sigma$.
        \item $\textsf{stable}$ is a function from $\Sigma$ to Booleans.
        \item The following axioms are satisfied:
        \begin{enumerate}
            \item $\forall x \in \Sigma \ldotp x = x \oplus |x| \land |x| = |x| \oplus |x|$
            \item $\forall x, c \in \Sigma \ldotp x = x \oplus c \Longrightarrow (\exists r \in \Sigma. |x| = c \oplus r)$
            \item $\forall a, b, c \in \Sigma \ldotp c = a \oplus b \Longrightarrow |c| = |a| \oplus |b|$
            \item $\textsf{stable}(e) \land (\forall a, b, c \in \Sigma \ldotp
            c = a \oplus b \land \textsf{stable}(a) \land \textsf{stable}(b) \Longrightarrow \textsf{stable}(c))$
            \item $\forall a, b, c \in \Sigma \ldotp c = a \oplus b \land c = c \oplus c \Longrightarrow a = a \oplus a$ \emph{(positivity)}
            \item $\forall a, b, x, y \in \Sigma \ldotp a = b \oplus x \land a = b \oplus y \land |x| = |y| \Longrightarrow x = y$
            \emph{(cancellativity)}
        \end{enumerate}
    \end{enumerate}
\end{definition}

$|\sigma|$ represents the pure (duplicable) resources contained in the state $\sigma$.
A state $\sigma$ is \emph{pure} iff $\sigma = \sigma \oplus \sigma$, that is, iff it contains only pure information.
In an implicit dynamic frame setting, the values stored in the heap are considered pure resources (and thus duplicable),
but ownership is not (since it cannot be duplicated).
Pure resources can also include local variables, which can be used 
to represent SL assertions of the form $\exists v \ldotp \pointsto{x.f}{v} * A$
as separating conjunctions without existentials (where the value of $v$ is represented as a duplicable resource).

In an implicit dynamic frame setting, a state might contain a value of a heap location even if it does not have permission to this heap location.
This is necessary to define the evaluation of the separating conjunction, as explained in \secref{subsec:idf_background}.
A state $\sigma$ is stable, written $\textsf{stable}(\sigma)$, iff $\sigma$ only contains values of heaps locations to which it has some permission.
When packaging a wand in IDF, we only consider stable footprints.
Not doing so would give a different semantics to IDF wands and SL wands.
Consider for example the SL assertion $A := \pointsto{\code{x.f}}{5} * (\pointsto{\code{x.f}}{\_} \wand \pointsto{\code{x.f}}{5})$, which can also be expressed in IDF~\cite{ParkinsonSummers12}.
In SL, $A$ is equivalent to false.
In IDF, $A$ is equivalent to false only if footprints are enforced to be stable.
Indeed, if we would allow non-stable states to be footprints of wands, then a state with no permission at all but in which \code{x.f} contains the value $5$
would be a valid footprint of $\pointsto{\code{x.f}}{\_} \wand \pointsto{\code{x.f}}{5}$.
Thus, a state in which $\pointsto{\code{x.f}}{5}$ holds would satisfy $A$.

Axioms (a) and (b) state that the core of a state is its maximal pure part, while axiom (c) requires the function $|\_|$ to be linear.
Axiom (d) requires that the unit is stable, and that the sum of two stable states is also stable.
Finally, the positivity axiom states that any state smaller than a pure state has to be pure,
and the cancellativity axiom states that the algebra is cancellative for non-pure resources.

Using this separation algebra, we define the following partial order
on elements of $\Sigma$:
A state $\sigma_2 \in \Sigma$ is larger than another state $\sigma_1 \in \Sigma$, written $\sigma_2 \succeq \sigma_1$,
iff $\exists \sigma_r \in \Sigma \ldotp \sigma_2 = \sigma_1 \oplus \sigma_r$.
Moreover, we write $\sigma_1 \# \sigma_2$ iff $\sigma_1 \oplus \sigma_2$ is defined.
Finally, we define a subtraction operator, $\sigma_A \ominus \sigma_B$, which corresponds to the largest state $\sigma_r$ such that
$\sigma_A = \sigma_B \oplus \sigma_r$ if $\sigma_A \succeq \sigma_B$ (the other case is not relevant).

In order to enable a package algorithm to deconstruct the separating conjunctions and implications of the right-hand side of a wand, and to extract the footprint
piecewise as illustrated in \appref{app:example-sound-package}, we consider an assertion language that contains the separating conjunction and the implication connectives.
This allows us to write logical rules that only apply to a separating conjunction and to an implication, respectively.
Moreover, to be as general as possible, we do not fix the other connectives of the assertion language.
Thus, the third type of assertion we consider in our assertion language is the general type of \emph{semantic assertions}, i.e. functions from $\Sigma$ to Booleans.
This third type represents any SL assertion that is neither a separating conjunction nor an implication.
In particular, assertions such as \code{x.f = 5}, \code{acc(x.f)}, abstract predicates (such as \code{Tree(x)}) or magic wands are represented as semantic assertions.

\begin{definition}
    Let $\mathcal{B}$ range over semantic assertions, i.e. functions from $\Sigma$ to booleans.
    Assertions (ranged over by $A$) are defined as follows:
    $$A = A * A \mid \mathcal{B} \Rightarrow A \mid \mathcal{B}$$

    For a state $\sigma \in \Sigma$ and an assertion $A$, we write $\sigma \models A$ to say that $\sigma$ satisfies $A$,
    and define it as follows:
    \begin{alignat*}{2}
        &\sigma \models A_1 * A_2 &&\Longleftrightarrow (\exists \sigma_1, \sigma_2 \ldotp
        \sigma = \sigma_1 \oplus \sigma_2 \land \sigma_1 \models A_1 \land \sigma_2 \models A_2) \\
        &\sigma \models \mathcal{B} \Rightarrow A &&\Longleftrightarrow (\mathcal{B}(\sigma) \Longrightarrow \sigma \models A) \\
        &\sigma \models \mathcal{B} &&\Longleftrightarrow \mathcal{B}(\sigma)
    \end{alignat*}
\end{definition}

\appref{app:extensions} explains how this extend this assertion language and the logic to handle other connectives, such as the disjunction or the normal conjunction.

This assertion language is too permissive for our purpose.
In particular, we only want to consider assertions that are well-formed, that is, if an assertion is well-defined and holds in a state,
then adding pure resources to this state should not render the assertion false.
Informally, for an IDF assertion, well-formed corresponds to being self-framing.
We achieve this with monotonicity constraints:
A semantic assertion $\mathcal{B}$ which appears on the left-hand side of an implication should stay false if we add pure resources to a state in
which it is false, and semantic assertions which are not on the left-hand side of an implication should behave in the opposite way.

\begin{definition}\label{def:wf-assertion}
    A semantic assertion $\mathcal{B}$ is \emph{monotonically pure}, written $\mathit{monoPure}(\mathcal{B})$,
    iff 
    $(\forall \sigma, \sigma_p \in \Sigma \ldotp
    \sigma_p \text{ is pure} \land
    \mathcal{B}(\sigma) \land
    \sigma \# \sigma_r \Longrightarrow \mathcal{B}(\sigma \oplus \sigma_p))$.
    
    We write $\mathit{wf}(A)$ to say that the assertion $A$ is \emph{well-formed}. It is defined as follows:
    \begin{alignat*}{2}
        &\mathit{wf}(A_1 * A_2) &&\Longleftrightarrow \mathit{wf}(A_1) \land \mathit{wf}(A_2) \\
        &\mathit{wf}(\mathcal{B} \Rightarrow A) &&\Longleftrightarrow
        \mathit{monoPure}( \lnot \mathcal{B})
        \land \mathit{wf}(A) \\
        &\mathit{wf}(\mathcal{B}) &&\Longleftrightarrow \mathit{monoPure}(\mathcal{B})
    \end{alignat*}
\end{definition}

\section{Extending the Logic}\label{app:extensions}

The framework and the logic presented in \secref{sec:package_logic} operate only on a simple language for assertions:
An assertion is either a separating conjunction (star) of two assertions, an implication of a pure semantic assertion on the left-hand side and an assertion on the right-hand side,
or a semantic assertion.
Note that any assertion can be represented in this framework, since any assertion can be represented as a semantic assertion.
The star and the implication connectives that our framework provides enables a package algorithm to (1) deconstruct an assertion with these connectives,
and (2) apply the rule $\mathit{Extract}$ with some heuristic at the ``leaves'' of this assertion.

Thus, while other connectives such as the disjunction or the non-separating conjunction can still be handled using semantic assertions,
one might want to extend the assertion language along with the logic such that the algorithm can deconstruct these assertions even deeper.
In the following, we describe how one can extend the set of rules from \figref{fig:package_rules}
to handle disjunctions and non-separating conjunctions.

\paragraph{Disjunctions.}
To satisfy the disjunction $A \lor B$, a state must either satisfy $A$ or satisfy $B$.
When dealing with a set of extended states, a package algorithm can choose which extended states must satisfy $A$, and which ones must satisfy $B$.
More precisely, a rule to handle disjunctions could proceed in five steps:
\begin{enumerate}
    \item Separate the witness set into the set $S_0^1$ of extended states that must prove $A$, and the set $S_0^B$ of extended states that must prove $B$.
    \item Use the rules to handle the assertion $A$ with the witness set $S_0^1$. This gives a new witness set $S_1^1$.
    \item In step 2, the algorithm might have added a partial footprint to the witness set $S_0^1$ to get the new witness set $S_1^1$.
    Thus, this partial footprint should be added to $S_0^B$, which gives a new witness $S_1^B$.
    \item Use the rules to handle the assertion $B$ with the witness set $S_1^B$, which gives a new witness set $S_2^B$.
    \item In step 4, the algorithm might have added a partial footprint to the witness set $S_1^B$ to get the new witness set $S_2^B$.
    Thus, this partial footprint should also be added to $S_1^1$, which gives a new witness $S_2^1$.
    \item The final witness set is $S_2^1 \cup S_2^B$.
\end{enumerate}

\paragraph{Non-separating conjunctions.}
The satisfy the non-separating conjunction $A \land B$, a state must satisfy $A$ and $B$.
Thus, the idea in this case is to first use the rules to satisfy $A$, then ``reset'' the states and use the rules to satisfy $B$,
and finally take the ``union'' of these states.
More precisely,
\begin{enumerate}
    \item Use the rules to handle the assertion $A$ with the initial witness set $S_0$. This gives a new witness set $S_1$.
    \item Record, for each extended state, the resources which have been added to its second element to go from $S_0$ to $S_1$,
        and then transfer back these resources from the second element to the first element of the extended state.
        This gives a new witness set $S_2$.
    \item Use the rules to handle the assertion $B$ with the initial witness set $S_2$. This gives a new witness set $S_3$.
    \item For each extended state of $S_3$, consider the state $\sigma_r$ which has been added to its second element to go from $S_2$ to $S_3$.
        If $\sigma_r \models A$, then do not modify this extended state.
        If $\sigma_r \not\models A$, then transfer another state $\sigma_r'$ from the first element to the second element,
        such that $\sigma_r \oplus \sigma_r' \models A$.
        These transformations yield the final witness set.
\end{enumerate}

\section{Example of a Wand that is not Combinable}\label{app:example-not-combinable}

In \secref{sec:combinable}, we show that the wand $\code{acc(x.f, 1/2)} \wand \code{acc(x.g)}$
is not combinable, because of a footprint that is incompatible with \emph{all} states that satisfy \code{acc(x.f, 1/2)},
but that becomes compatible with some when scaled down (by half).
In this section, we show a wand that is incompatible with \emph{some} states that satisfy the left-hand side of the wand, to illustrate that this is still an issue.

Consider, the wand $w' := \code{acc(x.f)} * (\code{x.f = y} \vee \code{x.f = z}) * \code{acc(x.f.g, 1/2)} \wand \code{acc(y.g)}$.
$w'$ is not combinable.
It is straighforward to see that $\code{acc(y.g)} \models w'$.
Moreover, $\code{acc(y.g, 1/2)} * \code{acc(z.g)} \models w'$.
Indeed, \code{acc(z.g)} combined with $(\code{x.f = y} \vee \code{x.f = z}) * \code{acc(x.f.g, 1/2)}$ implies that $\code{x.f} = \code{y}$,
and \code{acc(y.g, 1/2)} combined with $\code{x.f} = \code{y}$ and \code{acc(x.f.g, 1/2)} entails the right-hand side \code{acc(y.g)}.
However, $\code{acc(y.g)}^{0.5} * (\code{acc(y.g, 1/2)} * \code{acc(z.g)})^{0.5}
\equiv \code{acc(y.g, 3/4)} * \code{acc(z.g, 1/2)} \not\models w'$, since $\code{x.f} = \code{z}$ is now possible.
Footprints satisfying \code{acc(y.g)} are scalable.
However, footprints that only satisfy $\code{acc(y.g, 1/2)} * \code{acc(z.g)}$ are not scalable.

\section{A State Model for Fractional Permissions}\label{app:state_model}

We define in this section an implicit dynamic frame state model with fractional permissions,
to instantiate the separation algebra as described in \defref{def:separation-algebra}.
Moreover, we define the meaning of a \emph{binary} assertion.

\begin{definition}\textbf{State model}.\\
    Let $L$ be a set of heap locations which contains a special element $\mathit{null}$,
    and let $V$ be a set of values.
    
    A state is a pair $(\pi, h)$ of a permission mask $\pi$ and a partial heap $h$, where
    \begin{itemize}
        \item $\pi: L \rightarrow \mathbb{Q} \cap [0, 1]$ maps each heap location to a fractional permission between $0$ and $1$ included, and
        \item $h: L \rightharpoonup V$ is a \emph{partial} mapping from heap locations to values.
    \end{itemize}

    A state $(\pi, h)$ is \emph{valid} iff (1) $\pi(\mathit{null}) = 0$ and $\forall l \in L \ldotp \pi(l) > 0 \Longrightarrow h(l)$ is defined.
    (1) ensures that having ownership of a heap location implies that this heap location is not null, while
    (2) ensures that the values of all heap locations owned are defined.
    $\Sigma$ is defined as the set of all valid states.
\end{definition}

This state model corresponds to a separation algebra:

\begin{definition}
    Given two valid states $(\pi_1, h_1)$ and $(\pi_2, h_2)$, the addition $(\pi_1, h_1) \oplus (\pi_2, h_2)$ is defined iff
    (1) $h_1$ and $h_2$ agree on heap locations where they are both defined and
    (2) $\forall l \in L \ldotp \pi_1(l) + \pi_2(l) \le 1$.
    In this case, $(\pi_1, h_1) \oplus (\pi_2, h_2) = (\pi_1 + \pi_2, h_1 \cup h_2)$ is a valid state.

    The empty state $e$ is defined as $(\lambda \_ \ldotp 0, \emptyset)$.
    The \emph{core} of a state $(\pi, h)$ is defined as $|(\pi, h)| = (\lambda \_ \ldotp 0, h)$.
    A state is \emph{stable} iff $\forall l \in L \ldotp \pi(l) > 0 \Longleftrightarrow h(l) \text{ is defined}$.

    $(\Sigma, \oplus, e, |\_|, \textsf{stable})$ defines a separation algebra.
\end{definition}

We define the partial multiplication of a state by a positive rational as follows:
\begin{definition}
    Let $\alpha \in Q^+$ be a positive rational, and $(\pi, h) \in \Sigma$ be a valid state.
    The product $\alpha \odot (\pi, h)$ is defined iff $\forall l \in L. \alpha \times \pi(l) \le 1$.
    In this case, $\alpha \odot (\pi, h) := (\lambda l \ldotp \alpha \times \pi(l), h)$ (which is a valid state).
\end{definition}

\begin{definition}
    We define the \emph{binary restriction} of a permission mask $\pi$ as follows:
    $\mathit{bin}(\pi)(l) =
        \begin{cases}
          1 & \text{if } \pi(l) = 1 \\
          0 & \text{otherwise}
        \end{cases}$

    An assertion $A$ is \emph{binary} iff $\forall (\pi, h) \in \langle A \rangle \ldotp (\mathit{bin}(\pi), h) \in \langle A \rangle$
\end{definition}

\section{Leveraging the Logic}\label{app:lifting-logic}

The definition of the combinable wand $A \cwand B$ in \secref{sec:combinable} corresponds to the normal definition of a magic wand, except that the footprint is transformed before being combined with states
that satisfy $A$.
We generalise this pattern with the concept of \emph{monotonic transformers}.
A \emph{transformer} is a function $t$ that transforms a state $\sigma$ into the state $t(\sigma)$.
It is monotonic iff $\forall \sigma_1, \sigma_2 \ldotp \sigma_2 \succeq \sigma_1 \Longrightarrow t(\sigma_2) \succeq t(\sigma_1)$.
In the case of combinable wands, the function $\lambda \sigma \ldotp \mathcal{R}(\sigma_A, \sigma)$ is a monotonic transformer, for each $\sigma_A$ that satisfies $A$.
In the following, we explain how to lift the package logic such that it is sound and complete w.r.t. to the following wand's definition:
\[
    \sigma_w \models A \twand B \Longleftrightarrow \left( \forall \sigma_A \ldotp
    \sigma_A \models A \land \sigma_A \# \mathcal{T}(\sigma_A, \sigma_w) \Rightarrow \sigma_A \oplus \mathcal{T}(\sigma_A, \sigma_w) \models B
    \right)
\]
where $\lambda \sigma \ldotp \mathcal{T}(\sigma_A, \sigma)$ is a monotonic transformer for each $\sigma_A$ that satisfies $A$.
Note that we get the definition of the usual wand by setting $\mathcal{T}(\sigma_A, \sigma) = \sigma$ for all $\sigma_A$ and $\sigma$.

The witness set is lifted from a set of pairs of states $(\sigma_A, \sigma_B)$ to a set of tuples $(\sigma_A, \sigma_B, t)$, where
$t$ is the monotonic transformer associated to $\sigma_A$.
The initial witness set is thus $\{ (\sigma_A, e, \lambda \sigma \ldotp \mathcal{T}(\sigma_A, \sigma)) \mid \sigma_A \models A \} $.
We also need to modify the rule \emph{Extract}, such that we combine a \emph{transformed} version of $\sigma_w$ to elements of the witness set
(recall that $\sigma_w$ represents the permissions we extract from the outer state).
Consider a triple $(\sigma_A, \sigma_B, t)$ from the witness set.
We cannot simply combine $\sigma_A$ with $t(\sigma_w)$, because the footprint might be extracted piecewise in the package logic,
and the transformer $t$ is only applied to the \emph{complete} footprint in the above definition.
Therefore, we need to compute the \emph{part} of the transformed footprint that we need to combine with $\sigma_A$.
To do this, we need to keep track of the current footprint that has been extracted so far.
If $\sigma_f$ is the footprint extracted so far, and $\sigma_w$ is the additional part we want to extract from the outer state,
the state $\sigma_A$ must be combined with $t(\sigma_f \oplus \sigma_w) \ominus t(\sigma_f)$.
We subtract $t(\sigma_f)$ from $t(\sigma_f \oplus \sigma_w)$, since $t(\sigma_f)$ has already been added to this tuple.
In order to keep track of the footprint $\sigma_f$ extracted so far,
we extend contexts from a pair of a program state $\sigma$ and a witness set $S$ to tuples $(\sigma, S, \sigma_f)$.
Finally, the current footprint is updated to be $\sigma_f \oplus \sigma_w$ in the rule \emph{Extract}.
We have proven in Isabelle/HOL that this lifted logic is sound and complete for the above wand's definition.

\paragraph{Examples.}
Consider again the standard wand $w := \code{acc(x.f, 1/2)} \wand \code{acc(x.g)}$ and the states $\sigma_f$ and $\sigma_g$, containing full permissions to only \code{x.f} and \code{x.g}, respectively.
As explained in \secref{sec:combinable}, $w$ is not combinable, because it holds in both $\sigma_f$ and $\sigma_g$, but not in $0.5 \cdot \sigma_f \oplus 0.5 \cdot \sigma_g$.

Consider now the combinable wand $w_c := \code{acc(x.f, 1/2)} \cwand \code{acc(x.g)}$,
which is combinable because \code{acc(x.g)} is combinable (\thmref{thm:combinable}).
$\sigma_g$ is a valid footprint of $w_c$.
To see this, consider a state $\sigma_A$ in which the left-hand side \code{acc(x.f, 1/2)} holds,
and with no permission to \code{x.g}.
By \defref{def:restricted-footprint}, $\mathcal{R}(\sigma_A, \sigma_g) = \sigma_g$,
and thus $\sigma_A \oplus \mathcal{R}(\sigma_A, \sigma_g) = \sigma_A \oplus \sigma_g \models \code{acc(x.g)}$.

A derivation corresponding to this footprint can be easily computed using the previously lifted package logic,
as (partially) shown below. To ease reading, we ignore the path condition (which is always true),
we write $\mathcal{R}_{\sigma_A} := (\lambda \sigma \ldotp \mathcal{R}(\sigma_A, \sigma))$,
and $S_A := \{ \sigma_A \mid \sigma_A \models \code{acc(x.f, 1/2)} \land \sigma_A \# \sigma_g \}$
for the set of all states $\sigma_A$ compatible with $\sigma_g$ that
satisfy\footnote{In the case of an intuitionistic logic, the initial witness set does not need to contain all states that satisfy the left-hand side $A$ of the magic wand; it is sufficient and sound if it contains a set of \emph{minimal} states that satisfy $A$. 
In this example, we abuse the notation $\sigma_A \models \code{acc(x.f, 1/2)}$
to express states with \emph{exactly} half permission to \code{x.f} and no permissions otherwise.}
\code{acc(x.f, 1/2)}.
To package the combinable wand $w_c$ in a state $\sigma_g \oplus \sigma_r$ with the footprint $\sigma_g$,
we need to find a derivation of
$$
\langle \scode{acc(x.g)}, \left(\sigma_g \oplus \sigma_r, e, \{ (\sigma_A, e, \mathcal{R}_{\sigma_A}) \mid \sigma_A \models \code{acc(x.f, 1/2)} \} \right) \rangle
    \rightsquigarrow (\sigma_r, \sigma_g, S')
$$
for some witness set $S'$.

\vspace{5mm}
\begin{adjustbox}{max width=\textwidth}
    \Inf[Extract]{
        \Inf[Atom]
        {\forall \sigma_A \in S_A \ldotp \sigma_A \oplus \sigma_g \succeq \sigma_g \land \sigma_g \models \code{acc(x.g)}}.
        { \{ (\sigma_A, \sigma_g, \mathcal{R}_{\sigma_A}) \mid \sigma_A \in S_A \} = 
        \{ ((\sigma_A \oplus \sigma_g) \ominus \sigma_g , e \oplus \sigma_g, \mathcal{R}_{\sigma_A}) | \sigma_A \in S_A \} }.
        {\langle \scode{acc(x.g)}, \left(\sigma_r, \sigma_g, \{ (\sigma_A \oplus \sigma_g, e, \mathcal{R}_{\sigma_A}) \mid
        \sigma_A \in S_A \} \right) \rangle
        \rightsquigarrow (\sigma_r, \sigma_g, \{ (\sigma_A, \sigma_g, \mathcal{R}_{\sigma_A}) \mid \sigma_A \in S_A \} )}
    }{\dagger}
    {\langle \scode{acc(x.g)}, \left(\sigma_g \oplus \sigma_r, e, \{ (\sigma_A, e, \mathcal{R}_{\sigma_A}) \mid \sigma_A \models \code{acc(x.f, 1/2)} \} \right) \rangle
    \rightsquigarrow (\sigma_r, \sigma_g, \{ (\sigma_A, \sigma_g, \mathcal{R}_{\sigma_A}) \mid \sigma_A \in S_A \} )}
\end{adjustbox}

\vspace{5mm}

We first (reading bottom up) apply the rule \textit{Extract} to extract $\sigma_g$ from the outer state (parameter $\sigma_w$ in \figref{fig:package_rules}).
Since the footprint extracted so far is initially $e$ (the empty state),
each state $\sigma_A$ in the witness set is combined with $\mathcal{R}(\sigma_A, e \oplus \sigma_g) \ominus \mathcal{R}(\sigma_A, e) = \mathcal{R}(\sigma_A, \sigma_g) \ominus e = \sigma_g$.
$\dagger$ represents the three other premises of the rule, namely $\mathit{stable}(\sigma_g)$, $\sigma_g = e \oplus \sigma_g$, and that 
$\{ (\sigma_A \oplus \sigma_g, e, \mathcal{R}_{\sigma_A}) \mid \sigma_A \in S_A \}$
corresponds to $\{ (\sigma_A, e, \mathcal{R}_{\sigma_A}) \mid \sigma_A \models \code{acc(x.f, 1/2)} \}$ where $\mathcal{R}(\sigma_A, e \oplus \sigma_g) \ominus \mathcal{R}(\sigma_A, e)$ is added to each $\sigma_A$.
Since \code{acc(x.g)} holds in $\sigma_A \oplus \sigma_g$ (for each $\sigma_A$ in which \code{acc(x.f, 1/2)} holds),
we then apply the rule \textit{Atom}
(we ignore the premise about $S_\bot$ since the path condition is always true)
to conclude the proof.

On the other hand, $\sigma_f$ is \emph{not} a footprint of $w_c$.
Indeed, consider a state $\sigma_A$ with the same value as $\sigma_f$ for \code{x.f} and in which \code{acc(x.f, 1/2)} holds.
Since $\sigma_A$ is compatible with $0.5 \cdot \sigma_f$, the second case of the definition of $\mathcal{R}$ (\defref{def:restricted-footprint}) applies,
and thus $\mathcal{R}(\sigma_A, \sigma_f)$ only has ($\min(1, 0.5) =$) $0.5$ permission to \code{x.f}.
Therefore, $\sigma_A \oplus \mathcal{R}(\sigma_A, \sigma_f)$ is defined, but does not satisfy \code{acc(x.g)},
and thus $w_c$ does not hold in $\sigma_f$.
Because the package logic is sound (\thmref{thm:soundness}), it is not possible to find a derivation in the logic to prove that $\sigma_f$ is a footprint of the combinable wand $w_c$.
If we try to construct a proof similar to the one for $\sigma_g$, the application of the rule \textit{Extract} to extract $\sigma_f$ would succeed,
and update each $\sigma_A$ in the witness set to $\sigma_A \oplus (\mathcal{R}(\sigma_A, e \oplus \sigma_f) \ominus \mathcal{R}(\sigma_A, e)) = \sigma_A \oplus \mathcal{R}(\sigma_A, \sigma_f)$.
However, since this updated state does not satisfy \code{acc(x.g)}, the application of the rule \textit{Atom} would not succeed.

\section{Automation}\label{subsec:automation}

\def\cLHS{\mathit{consLHS}}
\newcommand*{\tab}{\text{\quad}}
\def\ret{\texttt{return}~}
\def\assert{\texttt{assert}}
\newcommand{\match}[1]{\tab A \text{ is } #1 \rightarrow  \\ \tab \tab \ret}

\def\pRHS{\mathit{proveRHS}}

\def\ifAlg{\texttt{if }}
\def\fiAlg{\texttt{fi}}
\def\thenAlg{\texttt{ then}}

For the sake of presentation, we only discuss here the algorithm that packages standard wands.
It is straightforward to adapt it to compute combinable wands, by following the approach described in \appref{app:lifting-logic}.
\figref{fig:package-algorithm} presents, on a high-level, the algorithm we have implemented in Viper's VCG.
Viper's VCG uses a total-heap semantics of IDF~\cite{ParkinsonSummers12},
where Viper states (ignoring local variables) consist of a heap and a permission mask (mapping resources to the held ownership amounts).
The heap and the mask are represented in Boogie with maps.
Based on this representation of Viper states, we can represent sets of states (in the case of $\cLHS$ below)
and witness sets (in the case of $\mathit{handleProofScript}$ and $\pRHS$ below) with Boogie maps, which allows us to flexibly manipulate these sets.

\begin{figure}[t]

    \scriptsize

    \begin{minipage}[t]{.46\textwidth}
        \vspace{0pt}
        $\begin{array}{l}
            \mathit{package}(\sigma_0, A \wand B, ps) =\\
            \tab S_0 \leftarrow \{ (\sigma_A, e) \mid \sigma_A \in \cLHS( T_0, \top, A) \}\\
            \tab (\sigma_1, S_1) \leftarrow \mathit{handleProofScript}(\sigma_0, S_0, ps)\\
            \tab (\sigma_2, S_2) \leftarrow \pRHS((\sigma_1, S_1), \top, B)\\
            \tab \ret \sigma_2
        \end{array}$
    \end{minipage}
    \begin{minipage}[t]{.54\textwidth}
        \vspace{0pt}
        $\begin{array}{l}
            \cLHS(T, pc, A) =\\
            \match{A_1 * A_2} \cLHS(\cLHS(T, pc, A_1), pc, A_2)\\
            \match{b \Rightarrow A} \cLHS(T, pc \land b, A)\\
            \match{b \text{ where } b \text{ is pure}} \{ \sigma_A \mid \sigma_A \in T \land \sigma_A \models b \} \\
            \match{r \text{ where } r \text{ is a resource}} \{ \sigma_A \oplus R(\sigma_A, r) \mid \sigma_A \in T \land \sigma_A \# R(\sigma_A, r) \}
        \end{array}$
    \end{minipage}

\renewcommand{\match}[1]{\tab B \text{ is } #1 \rightarrow}

    \vspace{3mm}

    $\begin{array}{l}
        \pRHS((\sigma, S), pc, B) =\\
        \match{B_1 * B_2} \ret \pRHS(\pRHS((\sigma, S), pc, B_1), pc, B_2)\\
        \match{b \Rightarrow B} \ret \pRHS(T, pc \land b, B)\\
        \match{b \text{ where } b \text{ is pure}}
        \assert(\forall \sigma_A \in S^1 \ldotp pc(\sigma_A) \Rightarrow (\sigma_A \models b))
        ; \; \ret (\sigma, S) \\
        \match{r \text{ where } r \text{ is a resource}}\\
            \tab \tab \assert(\forall \sigma_A \in S^1 \ldotp \sigma_A \# \sigma \Rightarrow \sigma_A \oplus \sigma \succeq R(\sigma_A, r))
            ; \; (\sigma', S') \leftarrow (\sigma, S)
            \\
            \tab \tab \ifAlg{} \lnot (\forall \sigma_A \in S^1 \ldotp \sigma_A \succeq R(\sigma_A, r)) \thenAlg \\
                \tab \tab \tab \text{Compute a minimal } \sigma_w \text{ s.t. } \sigma \succeq \sigma_w \text{ and }
                \forall \sigma_A \in S^1 \ldotp \sigma_A \# \sigma_w \Rightarrow \sigma_A \oplus \sigma_w \succeq R(\sigma_A, r))\\
                \tab \tab \tab (\sigma', S') \leftarrow (\sigma \ominus \sigma_w, \{ (\sigma_A \oplus \sigma_w, \sigma_B) \mid (\sigma_A, \sigma_B) \in S \land (\sigma_A \oplus \sigma_B) \# \sigma_w \})\\
            \tab \tab \fiAlg \\
            \tab \tab \ret (\sigma', \{ (\sigma_A \ominus R(\sigma_A, r), \sigma_B \oplus R(\sigma_A, r)) \mid (\sigma_A, \sigma_B) \in S' \})
    \end{array}$

    \caption{High-level representation of the package algorithm we have implemented in Viper's VCG to compute standard wands.
    $\mathit{package}$ is the main function,
    $\cLHS$ constructs a set of minimal states that satisfy an assertion, and $\pRHS$ automates a proof search in the package logic.
    }
    \label{fig:package-algorithm}
\end{figure}

The main function, $\mathit{package}$, takes as input a program state $\sigma_0$, a wand $A \wand B$, and a \emph{proof script} $ps$.
\thibault{We ignore proof scripts here since they are orthogonal to the automation of the proof search,}
but we \gout{briefly}explain what they are in~\appref{app:proofscript}.
The $\mathit{package}$ function calls the function $\cLHS$, which creates a minimal set of states that satisfy $A$, to create the initial witness set $S_0$.
It then calls the function $\pRHS$, which automates a proof search in the package logic, to extract a footprint of $A \wand B$ from $\sigma_0$.
The $\mathit{package}$ function finally returns the program state $\sigma_2$, which corresponds to the state $\sigma_0$ to which a footprint of $A \wand B$ has been subtracted.
After the package algorithm has successfully executed, the verifier can record an instance of the wand $A \wand B$ in $\sigma_2$ to get the new program state.

The call $\cLHS(T_0, \top, A)$ constructs a set $T$ of minimal states that satisfy the assertion $A$.
$T_0$ represents a set of ``empty'' states, i.e. states with no permissions but a total heap.
The functions $\cLHS$ and $\pRHS$ work similarly to each other, traversing the assertion they receive as input.
In particular, both $\cLHS$ and $\pRHS$ pattern match the assertion, which gives rise to four cases.
If the assertion is a separating conjunction, both functions handle first the first conjunct and then the second conjunct.
If it is an implication, both functions syntactically conjoin the left-hand side of the implication to their path condition $pc$.
In the case of $\pRHS$, the separating conjunction case corresponds to the rule $\mathit{Star}$ from the package logic,
and the implication case to the rule $\mathit{Implication}$.

Finally, both functions distinguish pure assertions from assertions that correspond to resources.
Resource assertions in Viper correspond to permissions to heap locations (e.g. \code{acc(x.f)}), to predicates (e.g. \code{Tree(x)}), or to magic wands.
Pure assertions are assertions that do not contain resources, such as \code{x.f = 5}.
In the case of a pure assertion $b$, $\cLHS$ filters out the states that do not satisfy $b$, while $\pRHS$ asserts that all elements of $S^1$ satisfy $b$.
The latter corresponds to an application of the rule $\mathit{Atom}$.\footnote{In this case,
the witness set $S$ is not modified because, for each $(\sigma_A, \sigma_B) \in S$, $\mathit{choice}(\sigma_A, \sigma_B)$
(recall that $\mathit{choice}$ is a parameter of the rule $\mathit{Atom}$)
corresponds to pure resources (as defined in \appref{app:state_model}) that are already present in $\sigma_A$ and $\sigma_B$,
and thus $\sigma_A \ominus \mathit{choice}(\sigma_A, \sigma_B) = \sigma_A$, and $\sigma_B \oplus \mathit{choice}(\sigma_A, \sigma_B) = \sigma_B$.}

To handle resource assertions, we use the notation $R(\sigma_A, r)$, which corresponds to a minimal state
that satisfies the resource assertion $r$ in the state $\sigma_A$.
We need to evaluate $r$ in $\sigma_A$ because $r$ might be heap-dependent, for example $r$ could be \code{acc(x.f.g)} or \code{Tree(x.left)}.
In the case of a resource assertion $r$, $\cLHS$ combines all states of the set $T$ with a minimal state that satisfies $r$,
while $\pRHS$ applies the following strategy:
If all states of $S^1$ satisfy $r$ (in which case the if-branch is not entered), $\pRHS$ directly applies the rule $\mathit{Atom}$.
In this is not the case, then $\pRHS$ first applies the rule $\mathit{Extract}$
by computing a minimal state $\sigma_w$ to extract from the outer state $\sigma$, which corresponds to the if-branch,
and then the rule $\mathit{Atom}$.
Finally, the initial $\texttt{assert}$ statement checks that all states of $S^1$ combined with the outer state satisfy $r$,
which ensures that the following code corresponds to correct applications of rules from the package logic.

\section{Proof scripts}\label{app:proofscript}
A proof script is a program statement that\gout{contains hints to} helps the package algorithm \gout{find a proof in the package logic}\gaurav{infer or check a footprint}.
They are mainly useful when one must manipulate predicate instances or magic wand instances in order to infer a footprint, since complete automation in such cases is infeasible.
\gaurav{Both Viper and VerCors support proof scripts. Since Viper's and our package algorithm infer footprints, one must provide less elaborate proof scripts than for VerCors.}

A package algorithm executes all program statements in a proof script before considering the wand's right-hand side.
Executing a proof script is similar to justifying the right-hand side: permission from the wand's left-hand side or the current state must be potentially used to do so.

In our implementation, we support the following inductively defined proof scripts:
\begin{align*}
P = \; &\bassert{A} \mid \bfold{\code{Q(x)}} \mid \bunfold{\code{Q(x)}} \mid \bapply{A \wand A} \mid{}\\
    & P ; P \mid \code{if} (b) \; \{ P \} \;\code{else} \; \{ P \}
\end{align*}
where \code{Q(x)} is a predicate instace, $b$ is a boolean expression, and $A$ is an assertion.

The proof script $\bassert{A}$ forces the algorithm to justify assertion $A$, which can be used to direct the algorithm towards a specific footprint.
\bfold{\code{Q(x)}} forces the algorithm to justify the permissions in the body of \code{Q(x)} and to exchange them for the predicate instance \code{Q(x)}, while \bunfold{\code{Q(x)}} does the opposite.
\bapply{$A \wand B$} forces the algorithm to apply the wand.
Finally, proof scripts can be composed sequentially or one can be put under a conditional.

The \code{package} in line~\ref{line:leftLeaf_package_step} of~\figref{fig:leftLeaf} requires the following proof script for Viper and for our package algorithm
$$ \bfold{\code{Tree(y0)}}; \bapply{\code{Tree(y0)} \wand \code{Tree(x)}} $$
to successfully infer a footprint, where \code{y0} is an auxiliary variable that contains the value of \code{y} at the beginning of the loop iteration.
Note that this proof script does not explicitly specify the footprint. To execute the two operations, the package algorithms still must remove the necessary permissions from the left-hand side or the current state.
In this case, the first \bkeyword{fold} statement forces the algorithm to select the permissions for \code{y0} as part of the footprint and the \bkeyword{apply} statement forces the algorithm to select the applied wand instance as part of the footprint (where the left-hand side is obtained after executing the \bkeyword{fold} statement).

In VerCors, the proof script for the \code{package} in line~\ref{line:leftLeaf_package_step} additionally requires \bkeyword{assert} statements (called \code{use} statements in VerCors) for all the permissions in the footprint, since VerCors does not infer footprints.
Moreover, one needs to add the statement \bassert{\code{y0.left = y}}, since VerCors cannot infer that \code{y} must be the left node of \code{y0}.

\end{document}